\documentclass[reqno,11pt]{amsart}
\usepackage{amsmath, latexsym, amsfonts, amssymb, amscd}
\usepackage{amsthm}
\usepackage{graphicx}
\usepackage{anyfontsize}

\numberwithin{equation}{section}

\newtheorem{theorem}{Theorem}[section]
\newtheorem{lemma}[theorem]{Lemma}
\newtheorem{proposition}[theorem]{Proposition}

\newtheorem{remark}[theorem]{Remark}

\newcommand{\ind}{\mathbf{1}}

\newcommand{\R}{\mathbb{R}}

\newcommand{\cP}{{\ensuremath{\mathcal P}} }

\newcommand{\cL}{{\ensuremath{\mathcal L}} }

\DeclareMathSymbol{\leqslant}{\mathalpha}{AMSa}{"36} 
\DeclareMathSymbol{\geqslant}{\mathalpha}{AMSa}{"3E} 
\DeclareMathSymbol{\eset}{\mathalpha}{AMSb}{"3F}     
\newcommand{\dd}{\,\text{\rm d}}             

\newcommand{\bbC}{{\ensuremath{\mathbb C}} }

\newcommand{\bbE}{{\ensuremath{\mathbb E}} }

\newcommand{\bbP}{{\ensuremath{\mathbb P}} }

\newcommand{\bbR}{{\ensuremath{\mathbb R}} }

\newcommand{\ga}{\alpha}
\newcommand{\gb}{\beta}
\newcommand{\gd}{\delta}
\newcommand{\gep}{\varepsilon}       

\newcommand{\go}{\omega}

\newcommand{\gs}{\sigma}

\makeatletter
\def\captionfont@{\footnotesize}
\def\captionheadfont@{\scshape}

\long\def\@makecaption#1#2{%
  \vspace{2mm}
  \setbox\@tempboxa\vbox{\color@setgroup
    \advance\hsize-6pc\noindent
    \captionfont@\captionheadfont@#1\@xp\@ifnotempty\@xp
        {\@cdr#2\@nil}{.\captionfont@\upshape\enspace#2}%
    \unskip\kern-6pc\par
    \global\setbox\@ne\lastbox\color@endgroup}%
  \ifhbox\@ne 
    \setbox\@ne\hbox{\unhbox\@ne\unskip\unskip\unpenalty\unkern}%
  \fi
  \ifdim\wd\@tempboxa=\z@ 
    \setbox\@ne\hbox to\columnwidth{\hss\kern-6pc\box\@ne\hss}%
  \else 
    \setbox\@ne\vbox{\unvbox\@tempboxa\parskip\z@skip
        \noindent\unhbox\@ne\advance\hsize-6pc\par}%
\fi
  \ifnum\@tempcnta<64 
    \addvspace\abovecaptionskip
    \moveright 3pc\box\@ne
  \else 
    \moveright 3pc\box\@ne
    \nobreak
    \vskip\belowcaptionskip
  \fi
\relax
}
\makeatother
\def\writefig#1 #2 #3 {\rlap{\kern #1 truecm
\raise #2 truecm \hbox{#3}}}

\newcommand{\mmu}{{\mathrm{m}}}

\newcommand{\vertiii}[1]{{\left\vert\kern-0.25ex\left\vert\kern-0.25ex\left\vert #1 
    \right\vert\kern-0.25ex\right\vert\kern-0.25ex\right\vert}}

\newcommand{\eps }{\varepsilon}
\renewcommand{\d}{\dd}
\newcommand{\gex}{{\hat{\gamma}_\eps}}

\DeclareMathOperator{\Log}{Log}

\newcommand{\res}{\textup{res}}
\renewcommand{\O}[1]{O\left({#1}\right)}
\newcommand{\Oeps}[1]{\O{\eps^{#1}}}

\newcommand{\E}[1]{\bbE \left[ #1 \right]}

\begin{document}

\title[ Lyapunov exponent of a product of random $2 \times 2$ matrices]
{Singular behavior of the leading Lyapunov exponent of a product of random $2 \times 2$ matrices}

\author[G. Genovese]{Giuseppe Genovese}
\address[Giuseppe Genovese]{Institut f\"ur Mathematik \\Universit\"at Z\"urich \\
Winterthurerstrasse 190 \\
CH-8057 Z\"urich \\ Switzerland}
\author[G. Giacomin]{Giambattista Giacomin}
\address[Giambattista Giacomin]{Universit\'e Paris Diderot, Sorbonne Paris Cit\'e,  Laboratoire de Probabilit{\'e}s et Mod\`eles Al\'eatoires, UMR 7599,
            F-75205 Paris, France}
\author[R. L. Greenblatt]{Rafael Leon Greenblatt}
\address[Rafael Leon Greenblatt]{Dipartimento di Scienze di Base e Applicate per l'Ingegneria \\
Universit\`a degli Studi di Roma ``La Sapienza''\\
Via Antonio Scarpa 14/16 \\
I-00161 Rome \\ Italy}
\curraddr{Institut f\"ur Mathematik \\Universit\"at Z\"urich \\
Winterthurerstrasse 190 \\
CH-8057 Z\"urich \\ Switzerland}
\thanks{
  The work of the first author is supported by Swiss National Science Foundation. The work of the third author was supported by Ricerche Universitarie Sapienza grant C26A14WR4N
and the PRIN National Grant \textit{Geometric and analytic theory of Hamiltonian systems in finite and infinite dimensions}.}
\email{rafael.greenblatt@gmail.com}
\date{\today}

\begin{abstract}
We consider a certain infinite product of random $2 \times 2$ matrices appearing in the  solution of some $1$ and $1+1$ dimensional disordered models in statistical mechanics, 
which depends on a  parameter $\varepsilon>0$ and on a real random variable with distribution $\mu$.  
For a large class of $\mu$, we prove the prediction by B.~Derrida and H.~J.~Hilhorst (J. Phys. A 16, 1641--2654 (1983)) that the Lyapunov exponent behaves like $C \eps^{2 \alpha}$ in the limit $\eps \searrow 0$, where $\alpha \in (0,1)$ and $C>0$ are determined by $\mu$.  Derrida and Hilhorst performed a two-scale analysis of the integral equation for the invariant distribution of the Markov chain associated to the matrix product and obtained a probability measure that is expected to be close to the invariant one for small $\gep$. We introduce suitable norms and exploit contractivity properties to show that 
such a probability measure is indeed close to the invariant one in a sense which implies a suitable control of the Lyapunov exponent. 
\end{abstract}

\maketitle

\section{Introduction}
\label{sec:intro}

\subsection{Products of random matrices, Lyapunov exponents and statistical mechanics}
Products of random $2 \times 2$ matrices have appeared in the physics literature since Schmidt \cite{Schmidt} introduced them to analyse a finite-difference equation with random coefficients proposed by Dyson \cite{Dyson} to study disordered harmonic chains. 
In the following years, probabilists and analysts began to investigate more general random matrix products, obtaining powerful results such as Furstenberg's theorem (regarding the existence and implicit characterization of the leading Lyapunov exponent \cite{F}) and Osoledets' multiplicative Ergodic theorem \cite{Os},  
many of which also hold in the more general context of linear cocycles (see~\cite{Viana} for a recent review).
The same difference equation studied by Dyson occurs as the Schr\"{o}dinger equation for the Anderson tight-binding model in one dimension, and Furstenberg's work played an important role in the first rigorous proofs of localization in this model (e.g.\ \cite{KS,Molchanov}); random matrix product theory provides a unified framework for these otherwise disparate treatments \cite{BL}.
Random-matrix-product studies of the one-dimensional Schr\"odinger equation have seen continued use in recent years to obtain further results about localization \cite{CTT10,CTT13}.

The present work considers random matrices of the form
\begin{equation}
\label{eq:M}
M_n^\gep\, :=\, 
 \begin{pmatrix}
1 & \gep \\
\gep Z_n & Z_n
\end{pmatrix}
\, ,
\end{equation}
where $\eps \in (0,1)$ is a constant and $\{Z_n\}_{n=1,2, \ldots}$ a sequence of positive, independent random variables with identical distribution  $\mu$. We will write $Z$ for a random variable with distribution $\mu$. This product  of random variables, and the associated Lyapunov exponent(s), appear in various statistical mechanics models. For example,
up to an unimportant factor, $M_n^\gep$ is the transfer matrix of the $1D$ Ising model with $\eps = e^{2 \beta J}$, $Z = e^{2 \beta h}$. Here randomness in $Z$ corresponds to a random magnetic field 
and the free energy density (in the thermodynamic limit) is the leading Lyapunov exponent \cite[Chapter 4]{CPV} defined by
\begin{equation}
\label{eq:Leps}
\cL(\gep)\,  :=\, \lim_{n \to \infty} \frac 1n\bbE 
\log \left\Vert  M_n^\gep M_{n-1}^\gep \cdots M_1^\gep \right\Vert,
\end{equation}
where $\| \cdot \|$ denotes an operator norm (the limit is independent of the norm chosen); 
our results apply to part of the frustrated regime, where the magnetic field can have either sign with nonzero probability.
Moreover, the free energy of the McCoy-Wu model in the thermodynamic limit can be expressed as an integral of the free energy of this model with respect to a parameter $q$ -- which maps to $\gep$ in our notation -- and the singular behavior comes from the values of $q$ (i.e.\ $\eps$) close to zero \cite{ShankarMurthy}.  
A similar matrix product also appears in the original treatment of the   McCoy-Wu model \cite{MW1}.  

\subsection{Working definitions and main result}
\label{sec:setup}

The classical theory of products of random matrices provides a technique for calculating $\cL(\eps)$ as an ergodic average \cite{BL,FK}.  
Since $\mathrm{det} (M_1^\gep) = (1-\gep^2) Z_1>0$,
\begin{equation}
\label{eq:A}
A_n^\gep \, := \, \frac{ M_n^\gep}{\sqrt{\mathrm{det}\left(  M_n^\gep \right)}}\, ,
\end{equation}
is a an element of $\mathrm{SL}(2, \bbR)$, and assuming that $\eps > 0 $ and that $\mu$ is absolutely continuous with bounded support, it is easy to confirm that it satisfies the assumptions of \cite[Chapter II, Prop.~4.3 and Th.~4.1]{BL}, which shows that the Markov process on $\cP(\bbR^2)$ defined by 
\begin{equation}
\label{eq:iteration123}
x, \ A_1^\gep x ,\  A_2^\gep A_1^\gep x , \  \ldots \,, \  A_n^\gep A_{n-1}^\gep \cdots A_1^\gep x, \
\ldots 
\end{equation}
has a unique (and therefore ergodic)
invariant probability measure $\mmu_\gep$.
As already remarked for example in  \cite{DH,MW1}, special features of the specific random matrices in question allow us to simplify the expression for $\cL(\gep)$.
Firstly, since all the matrix elements are positive and (for fixed $\mu$ and $\eps$, also assuming the support of $\mu$ is bounded away from 0) bounded from above and below, the limit is unchanged if we replace the vector norm $\| \cdot \|$ in the last line of \eqref{eq:Leps} with the scalar product with a fixed matrix element~\cite{FK}, i.e.\
\begin{equation}
\cL(\gep)\,  =\, \lim_{n \to \infty} \frac 1n\bbE 
\log \left[  M_n^\gep M_{n-1}^\gep \cdots M_1^\gep \right]_{11},
\end{equation}
and using the pointwise ergodic theorem we can rewrite this
\begin{equation}
  \begin{split}
	\cL(\gep)\, & =\, \lim_{n \to \infty} \frac 1n
	  \bbE \log \left[ A_n^\eps A_{n-1}^\eps \cdots A_1^\eps \right]_{11} 
	  + \frac 12 \bbE \, \log \mathrm{det}\left( M_1^\gep\right) 
	\\ & = \,
	\lim_{n \to \infty} \frac 1n \sum_{m=1}^n \bbE \left[ 
	  \log \frac{\left[ A_m x^{(m-1)} \right]_1}{x^{(m-1)}_1}
	  + \frac12 \log \mathrm{det}\left( M_m^\gep\right) 
	\right]
	\\ & = \,
	\int \bbE \log \frac{\left[M_1^\eps \hat x\right]_1}{\hat{x}_1}\mmu_\gep (\dd x)\, ,
  \end{split}
\end{equation}
where $x^{(m)} := A_m^\eps A_{m-1}^\eps \cdots A_1^\eps \begin{pmatrix}
  1 \\ 0
\end{pmatrix}$ and 
$\hat{x}$ is an arbitrarily chosen vector in $\mathbb{R}^2$ from the equivalence class $x$.

Secondly, since the elements of the first row of $M_1^\gep$ are deterministic, the first component of $M_1^\gep \hat{x}$ is a deterministic function of $\hat{x}$, and the $\bbE$ in the last line above is trivial.
Parameterizing $\cP(\bbR^2)$ by $\sigma \in (-\infty,\infty]$ with the choice $\hat{x} = (\eps, \sigma)$, for the $M$ under consideration we have explicitly
\begin{equation}
\label{eq:Leps2}
\cL(\gep)\, =\, 
\int \log (1 + \sigma) \, \bar{\omega}_\eps (\dd \sigma) \, ,
\end{equation}
where $\bar{\omega}_\eps$ is obtained from $\mmu_\eps$ by a change of variables; a simple computation shows that the Markov process defined by \eqref{eq:iteration123} corresponds to the one defined by the iteration
\begin{equation}
\label{eq:sigma_recursion}
\sigma_{n+1} = Z_n \frac{\eps^2 + \sigma_{n}}{1 + \sigma_n}\, ;
\end{equation}
$\bar{\omega}_\eps$ is then the unique stationary measury of this process.

All we have outlined up to now depends heavily on  $\gep\in (0,1)$. The case of $\gep\in (-1,0)$ can be dealt with just observing 
that $M_n^{-\gep}$ is conjugate to $M_n^\gep$ under the action of the diagonal matrix with eigenvalues $+1$ and $-1$, and therefore $\cL(\gep)=\cL(-\gep)$.
The case $\gep=0$ is however different: the invariant probability is not unique. In fact,  all the invariant probabilities can be written as convex combination of the 
Dirac deltas at $0$ and $\infty$, as can be seen by elementary arguments, and it is straightforward to see that $\cL(0)=\max(0,\bbE\log Z)$.

\medskip

Another fact that provides important context for our result is that, 
by applying the general result in \cite{Ruelle}, we see that 
 $\cL(\gep)$ is real analytic  (see also \cite{Hennion,LePage} that show $C^\infty$ behavior and H\"older continuity).  
 The singular character of the matrices for $\gep=0$ is, as we just pointed out, apparent,   but sharp results on the behavior of $\cL(\gep)$ for $\gep$ approaching zero are lacking in the mathematical literature.
  However  Derrida and Hilhorst \cite{DH}
claim that when $\E{Z} > 1$ and $\E{\log Z} < 0$
\begin{equation}
\label{eq:DH}
\cL(\gep) \stackrel{\gep\searrow 0} \sim C_\mu \gep^{2\ga}\, ,
\end{equation}
where $C_\mu$ is a positive constant and $\ga\in (0,1)$ is the unique positive real solution of
\begin{equation}
\bbE \left[ Z^\ga \right]\, =\, 1\, .
\label{eq:A_def}
\end{equation}
Existence and uniqueness of $\ga$
follow from the convexity of  the function $\gb\mapsto \bbE[Z^\gb]$, which takes value one with derivative 
$\E{ \log Z}<0$ at $\beta=0$  and value 
$\E{ Z}>1$ at $\beta=1$.

\medskip

The main result of the present work is a proof of \eqref{eq:DH}, in the following form:

\medskip

\begin{theorem}
\label{th:main}
Let $\mathcal{L}(\eps)$ and $M_n^\eps $ be as defined above, with $Z$ satisfying
\begin{enumerate}
\item $\bbE [Z]>1$ and $\bbE [\log Z ]  <0$;
\item There exist $c_-$ and $c_+$ with $0< c_-<c_+< \infty$
such that  $\bbP(Z \in [c_-, c_+])\,=\,1$, and there is no smaller closed interval so that this is true;  
\item The distribution $\mu$ of $Z$ is absolutely continuous with respect to the Lebesgue measure, and its density is a continuously differentiable function.
\end{enumerate}
Then there exist $\varkappa>0$ and $C_\mu>0$ such that
\begin{equation} 
\cL(\gep)\, =\, C_{\mu} \gep^{2\ga} +O\left( \gep^{2\ga + {\varkappa}}\right)\, ,
\end{equation}
for $\gep\searrow 0$ and $\ga$ the positive real solution of~\eqref{eq:A_def}.
\end{theorem}
\medskip

Note that assumption (1) implies that $c_-<1<c_+$. 
As we shall see, $\varkappa$ can be expressed explicitly given some information about the complex roots of \eqref{eq:A_def}, and our expression for $C_\mu$ is 
in agreement with \cite{DH}. 

The difficulty in proving Theorem~\ref{th:main} comes of course  from the
implicit characterization of 
$\mmu_\gep$. 
Identifying $\mmu_\gep$ is identifying the invariant measure of a Markov chain 
which does not have any special properties which would allow an explicit expression.
What has been exploited in \cite{DH} for \eqref{eq:DH}
is, in a sense, the solvable character of the model for $\gep=0$, 
but this limit is singular and taking advantage of it is by no means trivial, as we shall now explain.

\subsection{The Derrida-Hilhorst approach}
\label{sec:DH}

Let us review now the main argument of \cite{DH}. In view of \eqref{eq:Leps2},
we look at the evolution of $\gs$ under the random iteration~\eqref{eq:sigma_recursion}.
Recalling that $Z$ is supported on $[c_-, c_+]$, we see that 
the mapping~\eqref{eq:sigma_recursion} takes $(-\infty,\infty]$ into $[0,\infty)$, and $[0,\infty]$ into $[c_- \eps^2, c_+]$, and therefore
\begin{equation}
\label{eq:support}
\bar{\go}_\gep \left( [c_- \eps^2, c_+]^{\complement} \right) \, =\, 0\, .
\end{equation}
Incidentally, it is straightforward to see also that $\bar \go_\gep$ has a density (e.g.\ by the argument used in Proposition~\ref{lem:nu0} below); we will always denote the density of a measure with the same symbol, i.e.\ $\bar \go_\gep (\dd \gs)=\bar\go_\gep(\gs) \dd \gs$.

\medskip

The Derrida-Hilhorst approach is based on a two regime argument:
\medskip

\begin{itemize}
\item [I.] In the limit  $\gep \searrow 0$, the random recursion \eqref{eq:sigma_recursion} takes the form
\begin{equation}
\label{eq:newmapgs1}
\gs \mapsto Z\frac \gs {1+\gs}\, ,
\end{equation}
which defines a Markov chain whose unique invariant probability is concentrated at zero, since $\sigma_n < Z_n \cdots Z_1 \sigma_0$ which, for $\bbE \log Z < 0$, converges almost surely to $0$. 
However this chain has other invariant measures which cannot be normalized.  If instead of considering the stationary {\it probability} measures $\bar{\omega}_\eps$ we consider the stationary measures $\go_\eps$ with the normalization $\go_\eps\left( (y,\infty) \right) = 1$ for some suitable fixed $y$, these should have a nontrivial limit given by one of the non-probability stationary measures of the limiting process~\eqref{eq:newmapgs1},
which we denote by $\omega_0$. 
That is, for $\gs>0$ fixed  and $\gep$ small,   $\bar{\omega}_\eps(\gs)$ should be close to $a(\gep)\go_0(\gs)$, for some positive $a(\gep)$, and $a(\gep)=o(1)$ because 
the limit for $\gep \searrow 0$ of $\go_\gep$ concentrates in zero
 (see Figure~\ref{fig:density} and \eqref{eq:claimDH0}-\eqref{eq:claimDH0.1}).
  For $\sigma \searrow 0$, the recursion~\eqref{eq:sigma_recursion} formally converges to the linear map
 $\gs \mapsto Z \gs$, and the asymptotic behavior of $\omega_0$ should match that of a stationary measure of this map.
\item [II.] Derrida and Hilhorst  then analyse $\go_\gep$ by blowing up the scale by $\gep^{-2}$. Namely, they consider 
\begin{equation}
\label{eq:s_recursion}
s \, \mapsto\,   Z \frac{1+s}{1+\eps^2 s}\, .
\end{equation}
with stationary probability denoted by 
$\nu_\gep$, which, by  \eqref{eq:support}, is supported in $[c_-, \gep^{-2}c_+]$.
By taking the $\gep\searrow 0$ limit we get to $s \mapsto Z(1+s)$, which is linear but not straightforward to analyse.
The claim for this chain is that it does have a unique invariant probability $\nu_0$, supported on $[c_-, \infty)$, 
whose tail behavior ($s$ large) can be understood by    studying the simpler map $s \mapsto Zs$. 
\end{itemize}

\begin{figure}
\centering
\includegraphics[width=11.5 cm]{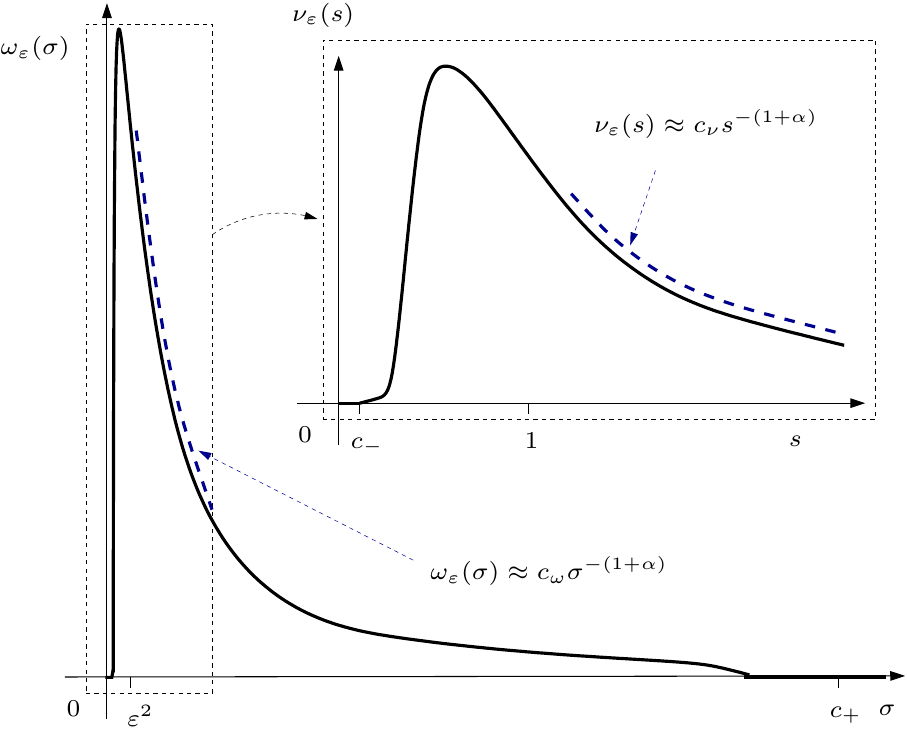}
\caption{\label{fig:density} 
A schematic view of the invariant measure at fixed small $\gep$. This plot shows the density $\gs \mapsto \go_\gep(\gs)$ (not a probability)    and the probability density
$s\mapsto \nu_\gep(s)$. The larger plot corresponds to regime I, while in the inset the horizontal coordinate  is blown up by a factor $\gep^{-2}$ and we are therefore in regime II. The {\sl intermediate regime} ($\gep^2 \ll \gs\ll 1$ or, equivalently,  $1 \ll s \ll \gep^{-2}$), where we expect to have  at the same time  $\go_\gep (s) \approx \go_0(s)$, $\nu_\gep(t)\approx\nu_0(t)$ and that the asymptotic expressions give good approximations, is highlighted by thick dashed curves. 
}
\end{figure}

\medskip

Therefore both the behavior  of $\go_0(\gs)$ for small $\gs$  and the behavior $\nu_0(s)$ for large $s$
are expected to be captured by the invariant measures of the random map on $(0, \infty)$ given by multiplication by $Z$. 
These measures have densities $f(\cdot)$ that satisfy 
$f(x)= \int_0 ^\infty f(y/x)\mu(y)y^{-1} \dd y$ and a simple computation shows that they can be written as
\begin{equation}
f(x)\, =\, \sum_i A_i x^{-1-\ga_i}\, ,
\end{equation}
and where $A_i\ge 0$ and $\ga_i\in \bbC$ are such that $\bbE[Z^{\ga_i}]=1$.
Derrida and Hilhorst  set forth  arguments suggesting that the asymptotic behaviors we are after are 
given by the case that contains only $\ga_i= \ga$, unique positive solution of $\bbE[Z^{\ga}]=1$, so 
\begin{equation}
\label{eq:CnuCgo}
\nu_0(s) \stackrel{s \to \infty}{\sim} \frac{c_\nu}{s^{1+\ga}} \ \ \text{ and } \ \ \ 
\go_0(\gs) \stackrel{\gs \searrow 0}{\sim} \frac{c_\go}{\gs^{1+\ga}}\, , 
\end{equation} 
where $c_\nu$ and $c_\go$ are positive constants: $c_\nu$ is fixed by the requirement that $\nu_0$ is a probability,
$c_\go$ is fixed by the normalization of the $\go_\gep$. Now the claim
is that
\begin{equation}
\label{eq:claimDH0}
\bar \go_\gep (\gs) \, \cong \, \begin{cases}
a(\gep) \go_0(\gs) & \text{ for } \gs\ge \gep\, ,
\\
\gep^{-2} \nu_0\left( \gep^{-2}\gs \right) & \text{ for } \gs< \gep\, ,
\end{cases}
\end{equation} 
where the symbol $\cong$ stands for {\sl approximately equal} and $a(\gep)$ can be evaluated by noting that the two terms should be essentially the same near $\gs=\gep$, which using \eqref{eq:CnuCgo} gives
\begin{equation}
\label{eq:claimDH0.1}
a(\gep) \stackrel{\gep \searrow 0} \sim \frac{c_\nu}{c_\go} \gep^{2\ga}\, .
\end{equation}  
Now we can go back to \eqref{eq:Leps2}, which can be computed using \eqref{eq:claimDH0}: for $\eps$ small,
\begin{equation}
\begin{split}
\cL (\gep) \,&= \, \int_{c_- \gep^2}^{\gep} \log(1+ \gs)  \ \bar \go_\gep(\gs) \dd \gs
+ \int_{\gep}^{c_+} \log(1+ \gs) \ \bar\go_\gep(\gs) \dd \gs
\\
&\cong \, C_\mu \gep^{2\ga} + \O{\gep^{1+\ga}},   \  \ \ \ \text{ with } 
C_\mu =\frac{c_\nu}{c_\go}  \int_{0}^{c_+} \log(1+ \gs) \ \go_0(\gs) \dd \sigma\, . 
\end{split}
\end{equation}

\subsection{Strategy of the proof and   structure of the paper} 
\label{sec:approach}
The strategy for our proof is very much inspired by \cite{DH}.  In short: we will construct a family of probability measures $\gamma_\eps$ by pasting together $\nu_0$ and $\omega_0$ (essentially, the right-hand side of \eqref{eq:claimDH0}) and use this to obtain an estimate of the Lyapunov exponent $\cL(\eps)$ which can be shown to be correct to within a remainder term which is of the order indicated in Theorem~\ref{th:main} as $\eps \searrow  0$. 

But this approach faces two main difficulties: the rigorous construction and analysis of  $\nu_0$ and $\go_0$ to 
get to a definition of $\gamma_\eps$ and (worse!) the fact in any case $\gamma_\eps$ is not the invariant probability. 
Let us elaborate on this: 

\medskip

\begin{enumerate}
\item To define the right-hand side of \eqref{eq:claimDH0} 
  one needs to construct and control $\nu_0$ and $\go_0$. This work is in part already done for $\nu_0$ \cite{deC}, but it is lacking for $\go_0$. Nevertheless, the road is paved for this analysis -- notably, we are going to use Mellin transform techniques similar to those of \cite{deC} --  and the difficulty that one needs to face are of technical nature.  
\item More substantially, $\gamma_\gep$ is not the invariant probability: we certainly expect it to be close to it, but in which sense and for what reasons? This step will be performed by introducing a family of norms that allows us to state in a precise and quantitative fashion that
$\gamma_\gep$  is one step {\sl approximately invariant}.  To use this to obtain meaningful estimates, we will show these norms are contracted by the action of the Markov Kernel; this provides an explicit estimate on the distance between $\gamma_\gep$ and the invariant probability $\nu_\gep$ and, more important 
for us,   the control is fine enough to pass to the functional of the invariant we are after: the Lyapunov exponent $\cL(\gep)$.
\end{enumerate}

\medskip

Now we are going to go more deeply into the strategy  and, in particular, we explain the tools and the fundamental ideas to deal with item (2) of the list. For this we find more practical to work in the scale of regime II, even if this is to a large extent arbitrary. 
We introduce the 
 map $T_\eps$ defined by
\begin{equation}
\int f(\tau)  T_\eps \nu (\d \tau) 
\, =\, 
\int \int f \left( t g_\eps(s)\right)  \nu (\d s) \mu (\d t)
\label{eq:T_def_integral}
\end{equation}
for all measurable bounded $f$ where
\begin{equation}
g_\eps(s) := \frac{1+s}{1+\eps^2 s}\, .
\label{eq:g_def}
\end{equation}
One readily checks  that $T_\gep$ is the one step transition map for the law of the chain
defined in \eqref{eq:s_recursion}, and hence $T_\gep \nu_\gep=\nu_\gep$.
\medskip

Later on we will need the analog of $T_\gep$, but in regime I.  So we introduce  $S_\gep$: 
 given a (finite) measure $\go$ on $(0, \infty)$ we write $S_\gep\go$ for the (dual) action of $S_\gep$ 
on $\go$, i.e. $S_\gep\go  (B)= \int S_\gep (\gs, B) \go(\dd \gs)$ for every Borel subset of $(0, \infty)$. 
Letting 
\begin{equation}
h_\eps(\sigma) := \frac{\eps^2 + \sigma}{1 + \sigma},
\label{eq:h_def}
\end{equation}
we have
\begin{equation}
\int f(\sigma)  S_\eps \omega (\d \sigma) 
=
\int \int f \left( z h_\eps(s)\right)  \omega (\d s)  \mu (\d z)\, ,
\label{eq:S_def_integral}
\end{equation}
for all measurable bounded $f$ and this is of course an alternative way to define $S_\gep\go$.

\medskip

We now introduce  for all measures $\nu$ supported on $[0,\infty)$ 
\begin{equation}
\label{eq:L[eps]}
L_\gep[\nu]\, :=\, \int \log(1+ \gep^2 s) \nu (\dd s)\,
\end{equation}
and observe  that $L_\eps[\nu_\eps] = \cL (\eps)$ (cf. \eqref{eq:Leps2})

\medskip

We also introduce the family of functionals on signed measures
\begin{equation}
\vertiii{\eta}_\beta \, :=\,  \int_0^\infty \tau^{\beta - 1} \left| G_\eta(\tau) \right| d \tau\, ,
\end{equation}
where we have used  the standard notation for the cumulative (tail) distribution $G$ and we take the occasion to 
introduce also the companion quantity $F$:
\begin{equation}
\label{eq:cumulative}
G_\nu (t) \, :=\, \nu ((t, \infty))
\ \ \ \text{ and } \ \ \
F_\nu (t) \, :=\, \nu ((-\infty,t])\, ,
\end{equation}
where, in general, $\nu$ is again a signed measure on $\bbR$.  Let us underline that, as is customary, by signed measure we mean a signed measure of finite total variation, that is the difference of two finite non negative measures. 
\smallskip

We have the following:
\medskip

\begin{lemma}
 \label{th:mainlem}
 For every $\gep \in(0,1)$, $\gb\in [0,1]$ and all probability measures $\nu_1$ and $\nu_2$,   supported in $[0, \infty)$, with  $L_\eps[\nu_2]$ and $L_\eps[\nu_1] $ finite,
we have
\begin{equation}
 \label{eq:mainlem2}
\left \vert L_\gep[\nu_1]- L_\gep[\nu_2] \right \vert\, \le \, \gep^{2\gb} \vertiii{\nu_1-\nu_2}_\beta\, .
\end{equation}
Moreover if $0< \gb < \ga$ there exists 
$c_\gb=c_\gb(\mu)$
independent of $\eps$
such that for every $\nu $ with $\vertiii{\nu}_\beta< \infty$
\begin{equation}
 \label{eq:mainlem3}
\vertiii{\nu_\gep - \nu}_\beta \, \le\,  c_\beta \vertiii{T_\gep \nu - \nu}_\beta\ .
 \end{equation}
 \end{lemma}
 
 \medskip

Lemma~\ref{th:mainlem} will be proven in Section~\ref{sec:contractivity} below.
We draw the attention of the reader on the fact that the important estimate \eqref{eq:mainlem3} is a contractivity property 
of $T_\eps$  with respect to $\vertiii{\cdot}_\beta$ for $0 < \beta < \alpha$. In fact 
\begin{equation}
\label{eq:Gstep1.s}
\vertiii{\nu_\gep - \nu}_\beta \, =\, \vertiii{T_\gep\nu_\gep - \nu}_\beta \, \le \, \vertiii{T_\gep (\nu_\gep - \nu)}_\beta +
 \vertiii{T_\gep \nu - \nu}_\beta\, ,
\end{equation}
from which it is apparent that \eqref{eq:mainlem3} follows once we have the claimed contractive property of  $T_\gep$.

It should also be quite clear at this stage that the key is to find a test measure $\nu$ that makes 
$\vertiii{T_\gep \nu - \nu}_\beta$ suitably small, so we can apply \eqref{eq:mainlem3} and then \eqref{eq:mainlem2},
with $\nu_1=\nu_\gep$ and $\nu_2=\nu$. The test measure of course corresponds to $\gamma_\gep$ presented informally in
Section~\ref{sec:DH}, although we will need to make a more precise definition, and we will find it convenient to do so in terms of distribution functions rather than densities since these appear more naturally in our arguments.
In any case, building $\gamma_\gep$ requires building first $\nu_0$ and $\omega_0$ and establishing properties of these 
two measures. This is done in Section~\ref{sec:omega0} and Section~\ref{sec:asymptotics}: we postpone the overview of these two sections  
and complete the argument that we are outlining.
In Section~\ref{sec:norm_bounds}  we   show in a rather straightforward way 
that
\begin{equation}
L_\eps[\gamma_\eps] \, =\,  C_\mu \eps^{2\alpha} + \Oeps{1+\alpha}\, ,
\end{equation}
and in Section~\ref{sec:ghat} we prove, on 
 the basis of a much less straightforward computation, 
 that
\begin{equation}
\label{eq:insecas}
\vertiii{T_\eps \gamma_\eps - \gamma_\eps}_\beta
\, =\, 
\Oeps{2\alpha - \beta}
+\Oeps{\alpha+\delta - \beta}
\end{equation}
for all $\beta \in (0,1)$ and for any $\delta>0$  chosen so that  $\bbE[Z^z]=1$, cf. \eqref{eq:A_def}, is solved in the strip  $\Re z \in [\alpha,\alpha+\delta]$
only by $z=\ga$ (as we shall see, the hypothesis that $\mu$ has a density largely suffices for the existence of such a $\gd$). Hence by using the two inequalities in
 Lemma~\ref{th:mainlem}
 we get to
\begin{equation}
\label{eq:Gstep1.f}
\begin{split}
  \left|
  \mathcal{L} (\eps) - L_\eps[\gamma_\eps]
  \right|
  \le 
  \eps^{2\beta} \vertiii{ \nu_\eps - \gamma_\eps}_\beta
  &\le 
  c_\beta \eps^{2\beta} 
  \vertiii{T_\eps \gamma_\eps - \gamma_\eps}_\beta
  \\& \ \ =
  \Oeps{2\alpha + \beta} + \Oeps{\alpha+\beta+\delta}\, ,
\end{split}
\end{equation}
and then choosing $\beta$ so that
\begin{equation}
\alpha > \beta > (\alpha - \delta) \wedge 0,
\end{equation}
we obtain Theorem~\ref{th:main}
with $\varkappa= \min(\beta, \beta+ \gd -\ga, 1-\ga) > 0$ .
This therefore concludes the argument.

\medskip

Let us spend a few words on the content of Section~\ref{sec:omega0} and Section~\ref{sec:asymptotics}:
in Section~\ref{sec:omega0} we  show the existence of the fixed points $\nu_0$ and $\omega_0$ as weak limits:
$\nu_\gep$ is a probability, but $\omega_0$ is instead an infinite measure which inherits the normalization $\go_\gep((y, \infty))=1$ for an appropriately chosen $y>0$. As we have already pointed out in Section~\ref{sec:DH},
the characterization  of $\nu_0$ (existence, uniqueness and asymptotic properties of $G_{\nu_0}$) is well known    \cite{kesten}, but we provide a (simple) proof for completeness and to introduce the methods used to show the existence of $\omega_0$.
In Section~\ref{sec:asymptotics} we characterize the behaviour of $G_{\nu_0}(s)$ for $s$ large and $G_{\omega_0}(s)$ for $s$ small, including some control on the subdominant terms (thereby improving on the result of \cite{kesten}, which only applies to $\nu_0$ and only gives the leading order behavior): an explicit control of the type $O(s^{-\gd})$ (for some $\gd>0$) on the ratio of the remainder to the leading term is crucial
 for our approach, as is clear from \eqref{eq:Gstep1.s}-\eqref{eq:Gstep1.f}.
 The proof is based on the characterization of the domain of analyticity of the Mellin transforms of $G_{\nu_0}$ and $G_{\omega_0}$.  As shown in \cite{deC}, the poles of the Mellin transform of $G_{\nu_0}(s)$ are either roots of \eqref{eq:A_def} or at integer translates of those roots (the latter are not important for our result); the relevant argument is summarized in Section~\ref{sec:nu_Mellin} for completeness.  
We  are not able to control the behavior of the Mellin transform well enough to use it to directly obtain an asymptotic expression for $G_{\nu_0}$ with control on the remainder, but we are instead able to do this (in Section~\ref{sec:nu_asymptotics}) at the level of the  primitive of  $G_{\nu_0}$ and then then recover the desired result on $G_{\nu_0}$ by reinjecting the estimate into 
the fixed point equation satisfied by $G_{\nu_0}$.
 We will also verify that there is a positive $\delta$ such that \eqref{eq:A_def} has no other roots with $\Re z \in [\alpha,\alpha+\delta]$, so that the subleading terms in the expansion are in fact smaller then the leading term by a factor $\O{s^{-\delta}}$.
In Section~\ref{sec:omega_asymptotics} we use the same techniques as in Sections~\ref{sec:nu_Mellin} and~\ref{sec:nu_asymptotics} to obtain similar results for the behavior of  $G_{\omega_0}$ near $0$.

\medskip

%

\subsection{Perspectives} 
\label{sec:perspectives}
Before embarking on the proof, we shall make a few remarks about the assumptions of Theorem~\ref{th:main} and perspectives for generalizations. 

\smallskip 

\begin{itemize}
\item
When \eqref{eq:A_def} has complex roots $a$ with $\Re a = \alpha$ (not covered by the present result: as shown in Section \ref{sec:asymptotics}, assumptions (2) and (3) of Theorem~\ref{th:main} ensure that complex roots have real part larger than $\ga$), the behavior in the intermediate regime can be given by a linear combination with the associated stationary measures, without violating the monotonicity of $G_{\nu_0}$.  
In this situation the leading Lyapunov exponent may instead behave like
\begin{equation}
\label{eq:DH_logper}
\cL(\gep) \stackrel{\gep\searrow 0} \sim \gep^{2\ga} H_\mu ( \log \eps)\, ,
\end{equation}
where $H_\mu $ is a nonconstant periodic function, which has been obtained by an exact calculation for a specific choice of $\mu$ in \cite{DH}.
\item
It is probably natural to expect that
Theorem~\ref{th:main} holds without conditions (2) and (3),  assuming instead  that there is a $\delta > 0$ such that \eqref{eq:A_def} has no complex solutions with $\Re a \in [\alpha,\alpha+\delta]$ (rather than deriving this using condition (3)), and that $\E{Z^a}<\infty$ for some $a > 0$.
Generalizing our approach in this fashion would at least complicate many of the estimates used
and require, in particular, a new approach to the  results in Section~\ref{sec:asymptotics} on asymptotic behavior of 
$G_{\omega_0}$ and $G_{\nu_0}$ would be needed.
If this could be done, the same methods might also be used to show that the log-periodic behavior in \eqref{eq:DH_logper} holds for other distributions $\mu$ besides the special case where it has been obtained so far \cite{deC,DH}. All of this however is not 
straightforward. 
\item
The cases excluded by assumption (1) are discussed in \cite{DH}: by replacing $Z$ with $1/Z$, as a corollary of Theorem~\ref{th:main} we obtain a similar result for distributions with $\E{\log Z} > 0$ and $\E{1/Z} > 1$.  The case $\E{\log Z} = 0$ remains of considerable interest, since it corresponds to the critical point of the statistical mechanical models discussed above, and obtaining more control over the behavior of $\mathcal{L}(\eps)$ as this condition is approached (along the lines of the discussion in \cite{ShankarMurthy}) appears to be worthwhile.  The case where $\E{Z} \le 1$, on the other hand, should merely exhibit a weakening of the singularity at $\eps = 0$.
\item
Finally, although we have focused on a concrete example of physical relevance, the method used here has the potential to generalize to matrices of other forms, most immediately other $2 \times 2$ matrices whose off-diagonal entries are $\Oeps{}$.
\end{itemize}

\section{Estimating the Lyapunov exponent with almost-stationary points of $T_\eps$}
\label{sec:contractivity}
In this section we will prove Lemma~\ref{th:mainlem}. The assertions in \eqref{eq:mainlem2} and~\eqref{eq:mainlem3} are separate.   
For brevity, we will take {\sl absolutely continuous} to mean absolutely continuous with respect to the Lebesgue measure on $(0,\infty)$.
Recall that we denote  absolutely continuous measures and their densities with the same symbol.

%

\medskip

\noindent
{\it Proof of \eqref{eq:mainlem2}.}

For any signed measure $\nu$ such that $G_\nu(x)$ is bounded and $L_\gep(\vert \nu \vert)< \infty$, we can integrate \eqref{eq:L[eps]} by parts to obtain
\begin{equation}
L_\eps[\nu] \, 
=\, 
\eps^2 
\int_0^\infty \frac{G_\nu (x) }{1+\eps^2 x} \d x\, . 
\end{equation}
This applies to $\nu_1 - \nu_2$, giving
\begin{equation}
\left|
L_\eps[\nu_1] - L_\eps[\nu_2]
\right|
\,=\,
\eps^2 \left| 
\int_0^\infty 
\frac{ G_{\nu_1}(x) - G_{\nu_2}(x) }{1+\eps^2 x}
\d x
\right|
\, \le\, 
\eps^2
\int_0^\infty 
\frac{ \left| G_{\nu_1}(x) - G_{\nu_2}(x) \right|}{1+\eps^2 x}
\d x \, .
\end{equation}
Noting that
\begin{equation}
\frac{1}{1+z} \,\le\,  z^{\beta-1}\, ,
\end{equation}
for all $\beta \in [0,1]$ and  $z >0$, this implies
\begin{equation}
\left|
L_\eps[\nu_1] - L_\eps[\nu_2]
\right|
\le \eps^{2\beta} \vertiii{\nu_1 - \nu_2}_\beta,
\end{equation}
that is, \eqref{eq:mainlem2}.

\qed

\medskip 

We now move to \eqref{eq:mainlem3}. The bulk of the proof is in the following lemma.

\medskip

\begin{lemma}
For any probability measures $\nu_1,\nu_2$ supported on $[0,\infty)$ such that
both $\vertiii{\nu_1}$ and $\vertiii{\nu_2}$ are finite 
 and any $\beta \in (0,1)$,
\begin{equation}
\vertiii{T_\eps \nu_1 - T_\eps \nu_2}_\beta 
\, <\,  
\bbE\left[ Z^\gb\right]
\vertiii{\nu_1 - \nu_2}_\beta\, .
\end{equation}
\label{lem:little_contractivity}
\end{lemma}

\begin{proof}
Choosing 
$f(x) = \ind_{(\tau, \infty)}(x)$,
\eqref{eq:T_def_integral} becomes
\begin{equation}
G_{T_\eps \nu} (\tau)
=
\int G_\mu\left( \frac{\tau}{g_\eps (s)} \right)
\nu ( \d s)\, . 
\label{eq:T_mu_CDF}
\end{equation}
Letting $\nu := \nu_1 - \nu_2$, we have $G_\nu(0) = G_\nu(\infty) = 0$.  
We can use this and the fact that $\mu$ is absolutely continuous to integrate by parts, obtaining
\begin{equation}
\begin{split}
\left|
G_{T_\eps \nu} (\tau)
\right|
&=
\left|
\int_0^\infty
\tau
\frac{1-\eps^2}{(1+s)^2}
\mu \left( \frac{\tau}{g_\eps (s)} \right)
G_\nu ( s) \d s
\right|
\\ &
\le 
\int_0^\infty 
\frac{\tau}{(1+s)^2}
\mu \left( \frac{\tau}{g_\eps (s)} \right)
\left|G_\nu ( s)\right| \d s
\end{split}
\end{equation}
so that
\begin{equation}
\vertiii{T_\eps \nu}_\beta
\le 
\int_0^\infty \frac{|G_\nu(s)|}{(1+s)^2} 
\int_0^\infty \tau^\beta \mu \left( \frac{\tau}{g_\eps(s)} \right)
\d \tau \d s\, .
\end{equation}
Noting that
\begin{equation}
\int_0^\infty \tau^\beta \mu \left( \frac{\tau}{g_\eps(s)} \right)
\d \tau\, 
=\,  \left( g_\eps (s) \right)^{\beta+1} 
\int_0^\infty z^\beta \mu(z) \dd z\, =\, \left( g_\eps (s) \right)^{\beta+1}  \bbE \left[ Z^\gb \right]\, ,
\end{equation}
and
\begin{equation}
\frac{\left( g_\eps (s) \right)^{\beta+1}}{(1+s)^2}
\, =\, 
\frac{(1+s)^{\beta - 1}}{(1+\eps^2 s)^{\beta+1}}\, <\, (1+s)^{\beta - 1}\, < \,  
 s^{\beta - 1}\, ,
\end{equation}
for $\beta \in (0,1)$ and $s>0$, we have 
\begin{equation}
\vertiii{T_\eps \nu }_\beta
\le \E{Z^\beta} \vertiii{\nu}_\beta \ .
\end{equation}
Since 
\begin{equation}
G_{T_\eps \nu_1}(\tau) - G_{T_\eps \nu_2} (\tau) 
= G_{T_\eps \nu}(\tau), 
\end{equation}
the conclusion follows immediately.
\end{proof}

\medskip

\noindent
{\it Proof of \eqref{eq:mainlem3}.}
To complete the proof of \eqref{eq:mainlem3},
we note that $\vertiii{\cdot}_\beta$ satisfies the triangle identity, and apply this along with Lemma~\ref{lem:little_contractivity} to obtain
\begin{equation}
\begin{split}
\vertiii{\nu_\eps - \nu}_\beta
=\,&
\vertiii{T_\eps \nu_\eps - \nu}_\beta
\le 
\vertiii{T_\eps \nu_\eps - T_\eps \nu}_\beta
+
\vertiii{T_\eps \nu - \nu}_\beta
\\ \le\,  &
\bbE\left[ Z^\gb \right]
\vertiii{\nu_\eps - \nu}_\beta
+
\vertiii{T_\eps \nu - \nu}_\beta\, .
\end{split}
\end{equation}
Then noting that $\bbE\left[ Z^\gb \right] < 1$ exactly for $\beta \in (0,\ga )$, we have the desired bound with
\begin{equation}
c_\beta \, =\,  \left( 
1 - \bbE\left[ Z^\gb \right] 
\right)^{-1}.
\end{equation}
\qed

\section{Existence of the limiting fixed points $\nu_0$ and $\omega_0$}
\label{sec:fixed_points}
\label{sec:omega0}

As alluded to in Section~\ref{sec:intro}, the definitions of the maps $T_\eps$ and $S_\eps$ are perfectly valid for $\eps = 0$, and their fixed points in this limiting case play an important role in our proof.  However, in this case the relationship of these operators to the random matrix product becomes singular, and we can no longer use the same techniques to establish the existence and uniqueness of these fixed points.

In this section, we use compactness and continuity arguments to establish these  results.  As already pointed out in the introduction,
$\nu_0$ has already been built and (partly) studied elsewhere. Our indirect approach to $\nu_0$, i.e. via $\nu_\gep$, may appear a bit convoluted 
since $\nu_0$ can be approached directly as the invariant probability of a Markov chain. We draw however the attention of the reader 
on the fact 
that standard 
approaches, like  Foster-Lyapunov criteria  (see e.g. \cite{MT}), 
are rather involved -- above all in the case of a non-countable state space -- and yield a lot of information that we do not need. That is, following \cite{MT}, we can find 
a Lyapunov function starting from the monotonicity properties of $x^\gb\ind_{(0, \infty)}(x)$, $\gb \in (0, \ga)$ under the action of the Markov kernel, as we do below for $\gep>0$; but then the completion of the proof 
requires verifying a  ``petite sets condition" (which is rather straightforward if one assumes that for every $\eta>0$ small 
$\min_{t\in [c_-+ \eta, c_+-\eta]}\mu(t)>0$, but in general becomes rather laborious) and the final result includes  uniqueness and (time!) mixing properties. 
The approach  in \cite{HM}, on the other hand, cannot in general be directly applied to our Markov kernel: some iterated version of the kernel should be used.
Once again the method also yields mixing. Our approach is not constructive and a priori it does not yield even 
uniqueness (for $\nu_0$ uniqueness is easily recovered, but uniqueness in reality is not even required for the rest of our proof to go thorough), but it is very concise, self-contained and it shoots for the information we really need. 
More importantly, a similar method also applies to the treatment of $\omega_0$, which is not a probability measure, and which is therefore not covered by the more usual techniques.

\medskip

Here and in the rest of this section all measures are Borel measures supported on $(0,\infty)$.

\begin{proposition}
\label{lem:nu0}
There is a probability measure $\nu_0$ such that $T_0 \nu_0 = \nu_0$,
and $\nu_0$ is absolutely continuous.
\end{proposition}

\begin{proof}
Recall \eqref{eq:support}.
We apply \eqref{eq:T_def_integral} to $f(x) = x^\beta$ to obtain 
\begin{equation}
\int \tau^\beta \nu_\eps (\d \tau)
=
\int \int
(t g_\eps(s))^\beta \nu_\eps ( \d s) \mu ( \d t)
=
\int t^\beta \mu ( \d t)
\int \left(g_\eps(s)\right)^\beta \nu_\eps (\d s).
\end{equation}
Noting that $(g_\eps (s))^\beta < (1+s)^\beta < 1 + s^\beta$ for $\beta \in (0,1)$ and $s, \eps > 0$, we see that
\begin{equation}
\label{eq:technu0}
\int \tau^\beta \nu_\eps (\d \tau )
<
\left(
\int z^\beta \mu ( \d z)
\right)
\left(
1+ \int \tau^\beta \nu_\eps(\d \tau)
\right)
\end{equation} 
which, for $\beta \in (0,\alpha)$, gives $\int \tau^\beta \nu_\eps (\d \tau) < k(\beta):= \bbE[Z^\gb]/(1-\bbE[Z^\gb])$ for all $\eps$.  Then by Markov's Inequality, we also have $G_{\nu_\eps}(x) < k(\beta) x^{-\beta}$, which implies that the $\nu_\eps$ are a tight family.  As a result, Prohorov's theorem implies that there is some sequence $\eps_n \to 0$ such that $\nu_{\eps_n}$ converges weakly.  Calling the limit $\nu_0$, we see that $\nu_0$ is a probability measure.
For readability, we will denote $T_n := T_{\eps_n}$, $\nu_n := \nu_{\eps_n}$, $g_n := g_{\eps_n}$ in the following.

We now need to confirm that $T_0 \nu_0 = \nu_0$.  We will do so by showing that $T_n \nu_n \to T_0 \nu_0$ weakly. We do this in two parts, by showing  that $T_0 \nu_n$ converges weakly to $T_0 \nu_0$ and that $F_{T_0\nu}(\tau)-F_{T_n\nu}(\tau)$ tends to zero for every $\tau$, uniformly in the choice of the probability measure $\nu$.

\begin{sloppypar}
Using \eqref{eq:T_def_CDF} we write
\begin{equation}
F_{T_0 \nu_0} (\tau) - F_{T_0 \nu_n} (\tau) 
=
\int
\left[
F_{\nu_0} \left( \frac{\tau}{g_0 (s)} \right)
- F_{\nu_n} \left( \frac{\tau}{g_0 (s)} \right)
\right]
\mu(\d s)\, .
\end{equation}
The integrand on the right hand side is bounded above by one and goes to zero \mbox{(Lebesgue-)}almost surely, hence also $\mu$ almost surely.  Then by dominated convergence the right hand side goes to 0 for all $\tau$, and indeed $T_0 \nu_n \to T_0 \nu_0$ weakly.
\end{sloppypar}

Then recalling \eqref{eq:T_mu_CDF}, for any probability measure $\nu$ we have
\begin{multline}
\left| 
F_{T_0 \nu} (\tau) - F_{T_n \nu}(\tau)
\right|
\, =\,
\int 
\mu
\left( \left[
\frac{\tau}{g_0(s)} , \frac{\tau}{g_n(s)}
\right) \right)
\nu (\d s)
\, 
\\ 
\le\, 
\tau
\left\|
 \mu 
\right\|_\infty
\left(
\frac{1}{g_n(s)}
- \frac{1}{g_0(s)}
\right)
\,
\le
\,  
\tau
\left\|
 \mu
\right\|_\infty
\eps_n^2,
\end{multline}
where we recall that $\Vert \mu\Vert_\infty$ is the maximum of the density of $\mu$ and 
 we have used the explicit definition~\eqref{eq:g_def} of $g_\eps(s)$ to obtain
\begin{equation}
\frac{1}{g_n(s)}
- \frac{1}{g_0(s)}
=
\eps_n^2 \frac{s}{1+s}
< \eps_n^2.
\end{equation}
Since this bound is uniform in $\nu$, it also implies that $T_0 \nu_n - T_n \nu_n \to 0$ weakly.  Together with $T_0 \nu_n \to T_0 \nu_0$, this implies that $T_n \nu_n \to T_0 \nu_0$, and since $T_n \nu_n = \nu_n \to \nu_0$ we see that $T_0 \nu_0 = \nu_0$.

This in turn implies that $\nu_0$ is absolutely continuous. Use the first identity in \eqref{eq:T_def_CDF}: noting that $g_0(s) = 1+s$ and 
\begin{equation}
\frac{\d}{\d \tau} 
F_\mu \left( \frac{\tau}{1+s} \right)\, 
=\,  \frac{1}{1+s} \mu \left( \frac{\tau}{1+s} \right)
\, \le\,   \| \mu \|_\infty
\, ,
\end{equation}
and that $\| \mu\|_\infty$ is finite (since the density $\mu$ is a continuous function with compact support),
then from \eqref{eq:T_mu_CDF} we see that $F_{T_0 \nu_0}$ is continuously differentiable with 
\begin{equation}
F_{T_0 \nu_0}'(\tau) 
= \int \frac{1}{1+s} \mu \left( \frac{\tau}{1+s} \right) \nu_0 (\d s)
< \infty.
\end{equation}
\end{proof}
\medskip

The situation will be similar for $\omega_0$, once we have fixed the normalization to obtain a nontrivial limit.  As a preliminary, 
\medskip

\begin{lemma}
\label{th:prelem}
For any $y < c_+-1$ and $\eps > 0$, 
$G_{\nu_\eps} (\eps^{-2} y)
> 0$ .
\end{lemma}
\medskip

\begin{proof} 
By \eqref{eq:T_def_CDF}
\begin{equation}
\label{eq:T_def_CDF2}
G_{ \nu_\gep}(\tau) 
\, =\,  
\int G_{\nu_\gep} \left( g_\eps^{-1} \left( \frac{\tau}{t} \right) \right) \mu (\d t)
-F_\mu(\gep^2 \tau)\, =\, \int_{\gep^2 \tau} ^\infty G_{\nu_\gep} \left( g_\eps^{-1} \left( \frac{\tau}{t} \right) \right) \mu (\d t)
\, .
\end{equation}
Set $b_\eps  := \inf \left\{ \tau :\,  G_{\nu_\eps} (\tau) = 0\right\}$.  By \eqref{eq:support}
we know  that $b_\gep \le \gep^{-2}c_+ $, but one can see also that $b_\gep \le \gep^{-2}(c_+ -\gd)$ for some $\gd\in (0, c_+)$.
Observe in  fact that \eqref{eq:T_def_CDF2} implies that 
\begin{equation}
G_{ \nu_\gep}(\gep^{-2}(c_+ -\gd)) 
\, =\,  
\int_{c_+ -\gd} ^{c_+} G_{\nu_\gep} \left( g_\eps^{-1} \left( \frac{\gep^{-2}(c_+ -\gd)}{t} \right) \right) \mu (\d t)
\, ,
\end{equation}
but $g_\eps^{-1} ( {\gep^{-2}(c_+ -\gd)}/{t} )\ge 
g_\eps^{-1}( {\gep^{-2}(c_+ -\gd)}/{c_+} )= (\gep^{-2}c_+/\gd)((1-(\gd/c_+))-\gep^2)$ for $t$ in the range of integration.
But since we know that $G_{\nu_\eps}(t)=0$  for every $t>\gep^{-2}c_+$ we obtain 
that 
$G_{ \nu_\gep}(\gep^{-2}(c_+ -\gd))$
is zero if $ (\gep^{-2}c_+/\gd)((1-(\gd/c_+))-\gep^2)> \gep^{-2}c_+$, that is if $\gd < (1-\gep^2) c_+/(c_+ +1)$.

We now claim that 
\begin{equation}
\label{eq:bgep_fp}
b_\gep\, =\, c_+ g_\eps(b_\eps)\, .
\end{equation}
In fact for $\tau>b_\gep$  we have $G_{\nu_\gep}(\tau)=0$ (we can choose $\tau< \gep^{-2}(c_+-\gd)$ for some $\gd >0$), which, by \eqref{eq:T_def_CDF2}, 
implies that $G_{\nu_\eps}(g_\eps^{-1}(\tau/t))=0$ for almost every $t>\gep^2\tau$ in the support of $\mu$, that is 
for almost every $t\in [c_-\vee \gep^2 \tau, c_+]$.
But since $G_{\nu_\eps}(\cdot)$ is non increasing, we have that $G_{\nu_\eps}(s)=0$ for every
$s>g_\eps^{-1}(\tau/c_+)$. This holds  for every $\tau\in(b_\gep, \gep^{-2}(c_+-\gd))$, so we obtain that
$G_{\nu_\eps}(s)=0$ for every
$s>g_\eps^{-1}(b_\gep/c_+)$, that is $g_\eps^{-1}(b_\gep/c_+)\ge b_\gep$.

On the other hand, if $\tau<b_\gep$ then $G_{\nu_\gep}(\tau)>0$, which requires 
$G_{\nu_\eps}(g_\eps^{-1}(\tau/t))>0$ for some $t\in [c_-\vee \gep^2 \tau, c_+]$, which in turn implies
that we have $G_{\nu_\eps}(g_\eps^{-1}(\tau/c_+))>0$, because $G_{\nu_\eps}(\cdot)$ is non increasing
and $g_\eps^{-1}(\cdot)$ is increasing. So for every $\tau<b_\gep$ we have $G_{\nu_\eps}(g_\eps^{-1}(\tau/c_+))>0$,
that is $g_\eps^{-1}(b_\gep/c_+)\le b_\gep$. Therefore \eqref{eq:bgep_fp} is established.

We can then explicitly solve  \eqref{eq:bgep_fp} to obtain
\begin{equation}
b_\eps \, =\, \frac{c_+ - 1 + \sqrt{(c_+-1)^2 + 4 \eps^2 c_+}}{2 \eps^2}
\,>\, 
\frac{c_+-1}{\eps^2}\, ,
\end{equation}
which completes the proof.
\end{proof}

We then fix a
\begin{equation}
\label{eq:y}
y\, \in\, \left(0, (c_+-1)\wedge \frac{c_+}2  \right)\, ,
\end{equation}
 and define $\omega_\eps$ by choosing the normalization
\begin{equation}
\label{eq:Gomy}
G_{\omega_\eps}(\sigma)
\,=\, 
\frac{G_{\nu_\eps}(\eps^{-2} \sigma)}{G_{\nu_\eps}(\eps^{-2} y)}\, ,
\end{equation}
so that $G_{\omega_\eps} (y) = 1$ for all $\eps > 0$.

We now consider the space of positive $\gs$-finite measures supported on $(0, \infty)$ equipped with 
the weak topology with respect to the bounded continuous functions whose support is bounded away from $0$,
that is the functions in $C^0_b((0, \infty); \bbR)$ that vanish in a neighborhood of zero.
We do so because, while $\go_\gep((0, \infty))< \infty$ 
as we will see 
the limit point has total mass $+ \infty$ because of a singular behavior at zero. 
\medskip

We have the following crucial estimate:
\medskip

\begin{lemma}
\label{th:powergrowth}
There exist $m \in(1, \infty)$, $a>1$ and a decreasing sequence $\{x_n\}_{n=0,1, \ldots}$ of positive numbers,
with $x_n \le y a^{-n}$, $y$ defined in \eqref{eq:y},  such that $G_{\omega_\eps} (x_n) \le m^n$ for all $n$ and all $\eps > 0$.
\end{lemma}
\medskip

\begin{proof}
To begin with, note that for any $\eps > 0$ and any $x,z > 0$,
\begin{equation}
\begin{split}
\label{eq:main74}
G_{\omega_\eps} (z) \, =\,  
\int G_\mu \left(\frac z {h_\eps(s)}\right) 
\omega_\eps (\d s)\, 
\ge\,  G_\mu \left(\frac z {h_\eps(x)}\right) G_{\omega_\eps}(x)\, 
\\ \ge\,  G_\mu \left(\frac {z(1+x)} {x}\right) G_{\omega_\eps}(x)\, ,
\end{split}
\end{equation}
where we have first used that both
since $G_\mu(\cdot)$ and $1/h_\eps(\cdot)$ are non increasing, so the composition of the two is non decreasing,  and then 
that $h_\eps (x ) > h_0(x)=x/(1+x)$ (and, again, the monotonicity of $G_\mu(\cdot)$).

\begin{sloppypar}
We now define a sequence $\{x_n\}_{n=0,1, \ldots}$ by setting 
$x_0=y$ and, for $n\ge 1$, ${x_n:=x_{n-1}/(k-x_{n-1})}$ for a  $k$ chosen in $(y+1\vee (y/2), c_+)$: note that $y+1\vee (y/2)< c_+$ by \eqref{eq:y}. 
Note also that the map $x\mapsto x/(k-x)$ has $0$ and $k/2$ as fixed points: they are both hyperbolic, $0$ is attractive while $k/2$ is repulsive. By \eqref{eq:y} we have $y<k/2$, hence 
$\{x_n\}_{n=0,1,\ldots}$ decreases  to zero exponentially fast: 
$x_n\le y/a^n$, $a:=k-y>1$. Moreover by \eqref{eq:main74} (with $z=x_{n-1}$ and $x=x_n$)
we have for $n=1, 2, \ldots$
\end{sloppypar}
\begin{equation}
\label{eq:1step}
G_{\omega_\eps}(x_n) \, \le\, G_{\omega_\eps}(x_{n-1})
  \left(
G_\mu 
\left(
x_{n-1} \frac{1+x_n}{x_n}
\right)
\right)^{-1}
\, ,
\end{equation}
so that by observing that $x_{n-1} (1+x_n)/x_n = k$, by   setting $m^{-1} := G_\mu (k)\in (0, 1]$
and recalling $G_{\omega_\eps}(x_{0})=1$ yields 
\begin{equation}
G_{\omega_\eps} (x_n) \, \le\,  m^n\, ,
\end{equation}
and the proof is complete.
\end{proof}

\medskip

Here is the main result about $\{\go_\gep\}_{\gep>0}$:
\medskip

\begin{proposition}
\label{th:omega0}
The family $\{\go_\gep\}_{\gep>0}$ is  compact and every limit point $\go_0$ satisfies $S_0 \omega_0 = \omega_0$, the support of $\go_0$
is in  $(0, c_+)$  and $G_{\omega_0}(y) =1$. Moreover 
there exists $U < \infty$ such that for every limit point $\go_0$ we have 
\begin{equation}
\int_0^\infty x^U \omega_0( \d x) \, =\, U \int_0^\infty x^{U-1}G_{\omega_0} (x)\d x< \infty\, .
\end{equation}
\end{proposition}
\medskip 

\begin{proof}
The compactness of $\{\go_\gep\}_{\gep>0}$ can be obtained by the Helly-Bray Lemma as follows: Consider in fact 
a decreasing sequence $\{x_n\}_{n=0,1,\ldots}$ of numbers in $(0, c_+)$ that tends to $0$. For every 
Borel subset $B$ of $\bbR$ we set
\begin{equation}
\go_{\gep, n}(B)\,:= \, \frac{\go_\gep(B\cap{(x_n, \infty)})}{M_n}\, ,
\end{equation}
with $M_n:= \max_\gep\go_\gep((x_n,\infty))$ and $M_n< \infty $  by Lemma~\ref{th:powergrowth}.
Therefore $\{ \go_{\gep, n}\}_{\gep>0}$ is a family of sub-probabilities; hence it is relatively compact 
by the Helly-Bray Lemma, so  also $\{M_n\go_{\gep, n}\}_{\gep>0}$ is relatively compact 
and of course $M_n\go_{\gep, n}(B)$ is just $\go_\gep(B \cap (x_n \infty))$. 
Via a diagonal  procedure we can therefore extract from any sequence tending to zero 
a  subsequence $\{\gep_j\}_{j=1, 2, \ldots}$ such that $\lim_j\int h(\gs) \go_{\gep_j} (\dd \gs )=\int h(\gs) \go_{0} (\dd \gs )$
for every bounded continuous $h$ whose support is bounded away from $0$. 
 
The proposed properties of $\go_0$ can now be confirmed directly, notably 
 $S_0 \omega_0 = \omega_0$ follows by the same argument used in Proposition~\ref{lem:nu0} with the obvious changes, in particular noting that 
\begin{equation}
\frac{1}{h_0(\sigma)} - \frac{1}{h_n(\sigma)}
\, =\,
\frac{\eps_n^2}{1+s}
\le \eps_n^2\, ,
\end{equation}
and $U$ is easily found by using Lemma~\ref{th:powergrowth}. Note in fact that Lemma~\ref{th:powergrowth}
implies the practical formulation 
\begin{equation}
G_\omega (x) \le
\left\{
\begin{array}{ll}
1, & x \ge y \\
m, & \frac{y}{a} \le x < y \\
m^2, & \frac{y}{a^2} \le x < \frac{y}{a} \\
\vdots & \vdots
\end{array}
\right.
\ \ \le 
\left\{
\begin{array}{ll}
1, & x \ge y \\
m \left( \frac{x}{y} \right) ^{-\log m / \log a}, &x < y
\end{array}
\right.
\end{equation}
and we see that any $U>\log m/ \log a$ will do: explicit values of the positive constant $a$ and $m$
are given in the proof of Lemma~\ref{th:powergrowth}.
\end{proof}

Finally, we note that
\begin{lemma}
$\omega_0$ defined in Proposition~\ref{th:omega0} is absolutely continuous.
\label{lem:omega_ac}
\end{lemma}
\begin{proof}
From \eqref{eq:S_def_integral}, we obtain
\begin{equation}
G_{S_\eps \omega} (\tau) = \int G_\mu\left( \frac{\tau}{h_\eps(s)} \right) \omega ( \d s)
\end{equation}
and, setting $\eps = 0$, $\omega = \omega_0$, this becomes
\begin{equation}
G_{\omega_0} (\tau) 
=
\int G_\mu \left( \frac{\tau}{h_0(s)} \right) \omega_0(\d s).
\end{equation}
Noting that
\begin{equation}
\frac{\d}{\d \tau} G_\mu \left( \frac{\tau}{h_0(s)} \right)
\, =\,
\frac{1}{h_0 (s)} \mu \left( \frac{\tau}{h_0(s)} \right)
\, \le\,  
\frac{c_+}{\tau} \| \mu \|_\infty \ind_{(h_0^{-1}(\tau/c_+),\infty)}(s)
\ind_{(0,c_+)}(\tau) \, ,
\end{equation}
we see that $G_{\omega_0}$ is differentiable on $(0,\infty)$ with the bound
\begin{equation}
G'_{\omega_0}(\tau)
=
\int
\frac{1}{h_0 (s)} \mu \left( \frac{\tau}{h_0(s)} \right) \omega_0 (\d s)
\le \frac{c_+}{\tau} \| \mu \|_\infty G_{\omega_0}\left(h_0^{-1}(\tau/c_+)\right)
\ind_{(0,c_+)}(\tau) < \infty \, .
\end{equation}
\end{proof}

\section{Asymptotic behaviour of $\nu_0$ and $\omega_0$} \label{sec:asymptotics}

This section culminates in Lemmas~\ref{lem:nu_asymptotics} and~\ref{lem:omega_asymptotics}, which will allow us to control the asymptotics of a number of integrals containing $G_{\nu_0}$ and $G_{\omega_0}$, which appear in Section~\ref{sec:norm_bounds} below.
We will do this by using Mellin transform techniques (see \cite{MellinChapter} for a review), which were applied to $\nu_0$ in~\cite{deC}; the calculation is reproduced in Section~\ref{sec:nu_Mellin}.

Generically, the poles of the Mellin transform of a function correspond to the powers appearing in its asymptotic expansion at $0$ or $\infty$; however, except when it is known \textit{a priori} that such an expansion exists, proving this requires some control on the behaviour of the Mellin transform for large imaginary argument.  No suitable control is available here, but we shall see that the expansion holds in a distributional sense which will be adequate for our purposes, as proven in Section~\ref{sec:nu_asymptotics}.

Finally, in Section~\ref{sec:omega_asymptotics}, we show that the techniques of the previous two sections can be adapted with minor changes to obtain similar results for $\omega_0$.

\subsection{Mellin analysis of $\nu_0$}
\label{sec:nu_Mellin}

Let 
\begin{equation}
{\mathtt M}(u) \,=\, \int t^u \mu(\d t)\, .
\end{equation}
As noted in Section~\ref{sec:approach}, we need $\delta > 0$ to be such that $\delta \le \alpha$, $\delta \le 1-\alpha$, and ${\mathtt M}(u) \neq 1$ for all complex $u$ with $\Re u \in [\alpha,\alpha+\delta]$ other than $u=\alpha$.  Let us confirm that there is actually a positive number satisfying these conditions.  As a preliminary, we will need the following version of the Riemann-Lebesgue lemma:
\begin{lemma}\label{lem:RL}
Let $\mu$ be a continuously differentiable function with support $[c_-,c_+] \subset (0,\infty)$, let $C \subset \R$ be a compact interval, and let $z_j$ be a sequence of complex numbers such that $|z_j| \to \infty$ and $\Re z_j \in C$.
Then
\begin{equation}
\int_{c_-}^{c_+} t^{z_j} \mu(t) \d t
\to 0. 
\end{equation}
\end{lemma}
\begin{proof}
Integrating by parts,
\begin{equation}
\left|
\int_{c_-}^{c_+} t^z \mu(t) \d t
\right|
=
\left|
\frac{1}{z+1}
\int_{c_-}^{c_+} t^{z+1} \mu'(t) \d t
\right|
\le 
\frac{1}{|z+1|}
\int_{c_-}^{c_+} t^{\Re z+1} \left| \mu'(t) \right| \d t
\end{equation}
and the last integral is bounded uniformly given $\Re z \in C$.
\end{proof}

\begin{lemma}
There is a $\delta > 0$ such that
\begin{equation}
\left\{
z \in \mathbb{C}
\,
\middle|\, 
\Re z \in [\alpha,\alpha+\delta] \textup{ and }
{\mathtt M}(z) = 1
\right\}
= \{\alpha\}.
\end{equation}\label{lem:delta}
\end{lemma}
\begin{proof}
%
For any $A > \alpha$, let $Z_A := \{ z \in \mathbb{C} | {\mathtt M}(z)=1 \textup{ and } \Re z \in (\alpha,A]\}$.
$Z_A$ cannot have any finite accumulation points, since ${\mathtt M}$ is entire and not constant, so if $Z_A$ is infinite it must also be unbounded.  
This would then imply the existence of a sequence $z_j$ such that $\Re z_j \in [\alpha,A]$, $| z_j| \to \infty$, and ${\mathtt M}(z_j) \equiv 1$, which is in contradiction with Lemma~\ref{lem:RL}.

We therefore know that $Z_A$ is finite for any $A > \alpha$, and since $\Re z > \alpha$ for any $z \in Z_A$ we conclude that $\inf_{z \in Z_A} \Re z - \alpha > 0$, and the result holds for any $\delta$ smaller than this quantity.
\end{proof}

Define a measure $\xi$ by
\begin{equation}
F_{\nu_0} (x ) = F_{\xi} (x+1),
\label{eq:xi_def}
\end{equation}
and let
\begin{equation}
{\mathtt X}(u) := \int t^u \xi (\d t).
\label{eq:X_def}
\end{equation}
${\mathtt X}$ is, up to a shift of the coordinate $u$ which we find convenient, the Mellin transform of $\xi$.  Among the calculations in \cite{deC}, we find the following:
\begin{lemma} \label{lem:X_mellin}
${\mathtt X}$ is analytic on the the set of complex numbers $u$ with $\Re u \le \alpha + \delta$, except for $u = \alpha$, which is a simple pole.
\end{lemma}

\begin{proof}
Since $\xi$ is a probability measure whose support is contained in $[1,\infty)$, it is clear that
\begin{equation}
\left| {\mathtt X}(u) \right| =
\left| \int t^{u} \xi (\d t) \right| \le \int t^{\Re u} \xi (\d t) \le 1
\label{eq:X_1bound}
\end{equation}
for all complex $u$ with $\Re u \le 0$; then on this region ${\mathtt X}$ is given by a uniformly absolutely convergent integral of an analytic function and is therefore itself analytic.  Similarly, recalling 
\begin{equation}
\label{eq:M(u)}
{\mathtt M}(u) \,=\,  \int t^u \mu(\d t),
\end{equation}
we see that ${\mathtt M}$ is an entire function (using the fact that the support of $\mu$ is bounded) with $|{\mathtt M}(u)| \le {\mathtt M}(\Re u)$. 

From the definition of $\xi$ as a translation of $\nu_0$, we see that \eqref{eq:T_def_integral} implies
\begin{equation}
\int f(t) \xi (\d t) = \int \int f\left( 1 + t z \right)  \mu(\d z) \xi (\d t)
\end{equation}
for all $\xi-$integrable $f$ (note that this is the stationarity condition of the iteration $t_{i+1} = 1+t_i z_i$, cf.\ \eqref{eq:s_recursion}).

Letting $f(t) = t^{u}$ we have
\begin{equation}
{\mathtt X}(u) = \int \int ( 1 + t z )^{u}  \mu(\d z)  \xi (\d t). \label{eq:xi_Mellin}
\end{equation}
 

The right hand side of \eqref{eq:xi_Mellin} can be rewritten using the Mellin-Barnes integral
\begin{equation}
(1 + tz)^{u} = \frac{1}{\Gamma(-u)}\frac{1}{2 \pi i }\int_{w_0 - i \infty}^{w_0 + i \infty} \Gamma(-w) \Gamma(w - u) (tz)^{w} \d w,
\end{equation}
valid for $\Re u < w_0 < 0$, giving us
\begin{equation}
{\mathtt X}(u) = \frac{1}{\Gamma(-u)}\frac{1}{2 \pi i }\int_{w_0 - i \infty}^{w_0 + i \infty} \Gamma(-w) \Gamma(w - u) {\mathtt X}(w) {\mathtt M}(w) \d w
\label{eq:X_countour}
\end{equation}
where we have used the exponential decay of the Gamma function in the imaginary direction to change the order of integration.  This decay is also more than enough to allow us to move the contour of integration.  For $\Re w \le 0$, the only poles of the integrand are those of $\Gamma(w-u)$, i.e.\ it has simple poles at $w = u,\  u - 1, \dots$.  The residue of $\Gamma(w-u)$ at $w=u$ is 1, so we have
\begin{equation}
{\mathtt X}(u) = {\mathtt X}(u) {\mathtt M}(u) + \frac{1}{\Gamma(-u)}\frac{1}{2 \pi i }\int_{w_0 - i \infty}^{w_0 + i \infty} \Gamma(-w) \Gamma(w - u) {\mathtt X}(w) {\mathtt M}(w) \d w
\end{equation}
for $\Re u -1 < w_0 < \Re u <0$, which we further rewrite as
\begin{equation}
{\mathtt X}(u) = \frac{1}{\Gamma(-u)[1-{\mathtt M}(u)]}\frac{1}{2 \pi i }\int_{w_0 - i \infty}^{w_0 + i \infty} \Gamma(-w) \Gamma(w - u) {\mathtt X}(w) {\mathtt M}(w) \d w.
\label{eq:X_recursion}
\end{equation}
As a function of $u$, the integral on the right hand side can be analytically continued into the  right half-plane, so long as we maintain the condition $\Re u - 1 < w_0 < \Re u$ and $w_0 < 0$ which prevents the contour of integration from encountering the poles of the Gamma functions.  
The right hand will have singularities only at the zeros of $1- {\mathtt M}(u)$, i.e.\ the solutions of \eqref{eq:A_def}.  
We have already seen that ${\mathtt X}$ is analytic for $\Re u \le 0$ (indeed the apparent singularity in \eqref{eq:X_recursion} at $u=0$ is removable, since $\Gamma(-u)$ also has a simple pole there).  
Since 0 and $\alpha$ are the only such zeros with $0 \le \Re u \le \alpha+\delta$, all that remains is to prove that $\alpha$ is not removable (it is obvious from \eqref{eq:X_recursion} that it is then a simple pole).

Suppose that ${\mathtt X}$ can be analytically continued at $\alpha$.  Then the Taylor series of ${\mathtt X}$ is absolutely convergent at some real $u$ with $u > \alpha$, giving
\begin{multline}
\sum_{n=0}^\infty \frac{(u - \alpha)^n}{n!} \int t^\alpha ( \log t)^n \xi (\d t) 
\,=\,  \int \sum_{n=0}^\infty \frac{(u - \alpha)^n}{n!}  t^\alpha ( \log t)^n \xi (\d t)
\,
\\ =\,  \int t^u \xi ( \d t) < \infty\, ,
\end{multline}
where the fact that all terms are positive on the support of $\xi$ has allowed us to exchange the sum and the integral.  
Then ${\mathtt X}(u)$ is given by a well-defined integral, and we can apply \eqref{eq:xi_Mellin} to obtain
\begin{equation}
{\mathtt X}(u) = \int \int (1+tz)^u \mu (\d z) \xi (\d t)
\ge {\mathtt M}(u) {\mathtt X}(u),
\end{equation}
which is impossible if ${\mathtt X}(u) \in (0,\infty)$, since $M(u) > 1$ for $u > \alpha$.  Since ${\mathtt X}(u)$ manifestly positive, we have obtained a contradiction.
\end{proof}

\subsection{Asymptotics of $\nu_0$}
\label{sec:nu_asymptotics}

Returning to the definition of ${\mathtt X}$,  \eqref{eq:X_def}, we can integrate by parts to obtain
\begin{equation}
{\mathtt X}(u) = u \int_1^\infty t^{u-1} G_\xi (t) \d t + G_\xi(1)
\end{equation}
for $\Re u \le 0$.  Noting furthermore that $G_\xi(1) = {\mathtt X}(0)$ 
we have
\begin{equation}
\frac{{\mathtt X}(u) - {\mathtt X}(0)}{u} 
=
\int_1^\infty t^{u-1} G_\xi (t) \d t, \label{eq:G_xi_Mellin}
\end{equation}
so that $[{\mathtt X}(u)-{\mathtt X}(0)]/u$ is the Mellin transform of $\ind_{[1,\infty)}G_\xi$ with fundamental strip containing $\Re u \le 0$.  
We can use the inverse Mellin formula to
convert this into a formula for $G_\xi$, noting that by \cite[Theorem 28]{Titchmarsh} this formula is valid since $G_\xi$ is a continuous function (by Proposition~\ref{lem:nu0}) with bounded local variation (being bounded and monotone).
This gives  
%
%
%
\begin{equation}
G_\xi (t) = \frac{1}{2 \pi i	} \int_{u_0 - i \infty}^{u_0 + i \infty}  t^{-u} \frac{ {\mathtt X}(u) - {\mathtt X}(0)}{u} \d u
\label{eq:G_inverse_Mellin}
\end{equation}
for all $u_0 \le 0$ and $t > 1$. 

To get asymptotics from this expression we will need to displace the contour of integration further to the right, which requires some control on the growth of ${\mathtt X}$.  We can obtain this from \eqref{eq:X_recursion} as follows:
\begin{lemma} \label{lem:X_growth}
For $|\Im u|$ large,
\begin{equation}
|{\mathtt X}(u)| = \O{|\Im u |^{1/2}|}
\end{equation}
uniformly for $u_0 \in [0,U]$ for any $U < 1$.
\end{lemma}

\begin{proof}
Fix some $w_0 \in (-1/2,0)$ so that $u_0 - 1 < w_0$.
Then recalling $|{\mathtt X}(w_0 + i x )| \le 1$ (see \eqref{eq:X_1bound}), \eqref{eq:X_recursion} yields
\begin{equation}
|{\mathtt X}(u)| \le 
\frac{|{\mathtt M}(w_0)|}{|1- {\mathtt M}(u)|} 
\frac{1}{|\Gamma(-u)|}
\frac{1}{2 \pi}
\int_{-\infty}^\infty \left|
\Gamma ( - w_0 - ix) \Gamma(w_0 + i x - u)
\right| \d x.
\label{eq:X_decay_bound1}
\end{equation}
Lemma~\ref{lem:RL} shows that ${\mathtt M}(u) \to 0$ as $\Im u \to \pm \infty$, and examining the proof we see that the convergence is in fact uniform, so the first ratio on the right hand side is uniformly $O(1)$.

Using Stirling's series and then the expansion $\arctan(x)-\pi/2+1/x=O(1/x^3)$
  for $x\to \infty$, we have
\begin{multline}
\log | \Gamma (z) | = \Re \Log \Gamma (z)
= \\
\frac12 \log 2\pi + \left( \Re z - \frac12 \right) \log |z| - \Re z - (\Im z) \arg z + \O{\frac{1}{\Im z}}
\\
= \frac12 \log 2\pi + \left( \Re z - \frac12 \right) \log |z|  - \frac{\pi}{2} |\Im z| + \O{\frac{1}{\Im z}}\, ,
 \label{eq:Stirling}
\end{multline}
 where the first expression holds for $\Im z$ large uniformly in $\Re z$ \cite{Spira} and the second 
 holds uniformly over compact sets.
Thus for $\Im u$ large, 
\begin{multline}
\int_{-\infty}^\infty \left|
\Gamma ( - w_0 - ix) \Gamma(w_0 + i x - u)
\right| \d x
=\\
\O{
\int_{-\infty}^\infty
|x|^{-w_0 - \tfrac12} 
\left| x - \Im u \right|^{w_0 - u_0 - \tfrac12}
\exp \left(
	- \frac{\pi}{2} \left[ |x| + |x - \Im u| \right]
\right)
\d x
}
\\=
\O{
	\exp \left(
		- \frac{\pi}{2} | \Im u|
	\right)
	|\Im u |^{-u_0}
}
\end{multline}
uniformly for $\Re u \in [0,U]$. 
Plugging this into \eqref{eq:X_decay_bound1} and estimating $\Gamma(-u)$ using \eqref{eq:Stirling}, we obtain the desired result.
\end{proof}

Along with Lemma~\ref{lem:X_mellin}, this allows us to displace the contour in \eqref{eq:G_inverse_Mellin} to obtain
\begin{equation}
G_\xi (t) =  \frac{\res_{\mathtt X} (\alpha) }{\alpha} t^{-\alpha}
 +  \frac{1}{2 \pi i	} \int_{u_0 - i \infty}^{u_0 + i \infty}  t^{-u} \frac{{\mathtt X}(u) - {\mathtt X}(0) }{u} \d u\, ,
\label{eq:G_xi_expansion}
\end{equation}
for all $t > 1$ and any $u_0 \in (\alpha,U]$ where $\res_{\mathtt X}(a)$ denotes the residue of ${\mathtt X}$ at $a$ and for some $U > \alpha + \delta$.  If the integral on the right hand side were absolutely convergent, it would be $O(t^{-\alpha - \delta})$.  
Unfortunately this is not the case; however all is not lost.  

Denoting the value of this integral by $R_\xi(t)$ and noting that it is independent of $u_0$ within the specified interval, for any $s \in (\alpha,U)$ we can choose $u_0 \in (\alpha,s)$ and obtain
\begin{equation}
  \begin{split}
    & \int_0^\infty t^{s-1} R_\xi (t) \d t
    \\
    & \ =
    \frac1{2\pi i} \bigg[ 
      \int_0^1 \int_{u_0 - i \infty}^{u_0 + i \infty}
      t^{s-u-1} \frac{{\mathtt X}(u) - {\mathtt X}(0) }{u}
      \d u \d t
      \\ 
  &    \phantom{movemovemovemove}
      + 
      \int_1^\infty \int_{U - i \infty}^{U + i \infty}
      t^{s-u-1} \frac{{\mathtt X}(u) - {\mathtt X}(0) }{u}
      \d u \d t
    \bigg]
    \\
    & \ = 
    \frac1{2\pi i} \left[ 
      \int_{u_0 - i \infty}^{u_0 + i \infty}
      t^{s-u} \frac{{\mathtt X}(u) - {\mathtt X}(0) }{u(s-u)}
      \d u \d t
      - 
      \int_{U - i \infty}^{U + i \infty}
      t^{s-u} \frac{{\mathtt X}(u) - {\mathtt X}(0) }{u(s-u)}
      \d u \d t
    \right]
    \\
    & \ =
    \frac{{\mathtt X}(s) - {\mathtt X}(0) }{s},
  \end{split}
  \end{equation}
noting that the $t$ integrals are uniformly convergent and recognizing the resulting expression as an integral around a closed contour.  
This establishes that the Mellin transorm of $R_\xi$ on the strip containing $\alpha+\delta$ is $\frac{{\mathtt X}(u) - {\mathtt X}(0) }{u}$, and \cite[Theorem~11.10.1]{ML.transforms} gives an expression for the unique distribution with that property:
\begin{equation}
  R_\xi (t)
  =
  \frac{\d^2 }{\d t^2} \left[ \frac{1}{2\pi i}
  	\int_{\alpha+ \delta - i \infty}^{\alpha+\delta + i \infty}
      		\frac{{\mathtt X}(u) - {\mathtt X}(0) }{u(u-1)(u-2)} t^{2-u}
		\d u
	      \right].
\end{equation}
The bounds on the growth of ${\mathtt X}$ form Lemma~\ref{lem:X_growth} are sufficient to take one of the derivatives inside the integral, giving $R_\xi(t) = Q_\xi'(t)$ where
\begin{equation}
Q_\xi (t) := \frac{1}{2 \pi i} \int_{\alpha + \delta - i \infty}^{\alpha + \delta + i \infty} \frac{{\mathtt X}(u) - {\mathtt X}(0)}{u (1 - u)} t^{1-u} \d u,
\label{eq:Q_xi_def}
\end{equation}
and Lemma~\ref{lem:X_growth} suffices for \textit{this} integral to be absolutely convergent.  
As a result there is a constant $C_\xi$ such that
\begin{equation}
\left|Q_\xi(t)\right| \le C_\xi t^{1-\alpha - \delta}
\label{eq:Q_xi_bound}
\end{equation}
for all $t > 0$.

Noting that the definition of $\xi$ is such that $G_{\nu_0} (t) = G_\xi ( t+1)$, can use this to obtain:

\medskip

\begin{lemma} \label{lem:nu_asymptotics}
There is a function $R_\nu:\mathbb{R}^+ \to \mathbb{R}$ and a constant $C_\nu \neq 0$ such that 
\begin{equation}
G_{\nu_0} (t) = C_\nu \left[t^{-\alpha} + R_\nu (t) \right]
\label{eq:Gnu_expansion}
\end{equation}
where 
\begin{equation}
	R_\nu ( t) = \O{t^{-\alpha - \delta}}
	\label{eq:Rnu}
\end{equation}
for $ t \to \infty$.
\end{lemma}

\medskip

\begin{proof}
	Taking \eqref{eq:Gnu_expansion} as a definition of $R_\nu$ with $C_\nu = \res_{\mathtt X} (\alpha) / \alpha$, then we obtain Equation~\eqref{eq:Gnu_expansion} with 
$R_\nu(t) = \frac{1}{C_\nu} Q'_\xi(t+1)$.
Then recalling that $G_{\nu_0} = G_{T_0 \nu_0}$ and writing out the action of $T_0$ as in Equation~\eqref{eq:alt-G}, we obtain
\begin{equation}
	t^{-\alpha} + R_\nu (t)
	=
	\int \left[ \left( \frac{t}{z} - 1 \right)^{-\alpha}
	+ Q'_\xi \left( \frac{t}{z} \right)\right]
	\mu (\d z) \ ,
	\label{eq:nu_stationary_expansion}
\end{equation}
where we have used $g_0^{-1}(x) = x-1$ and cancelled a factor of $C_\nu$.  Using the generalized binomial theorem,
\begin{equation}
	\left( \frac{t}{z} - 1 \right)^{-\alpha}
	= \left( \frac{z}{t} \right)^{\alpha} \left( 1 - \frac{z}{t} \right)^{-\alpha}
	= \left( \frac{z}{t} \right)^{\alpha} \sum_{k=0}^\infty \binom{-\alpha}{k} \left( \frac{z}{t} \right)^k,
\end{equation}
where the series on the right-hand side is absolutely convergent for $t > z$; then using the fact that $\mu$ is a probability measure with support $(c_-,c_+)$ and $\int z^\alpha \mu(\d z) = 1$,
\begin{equation}
	\int  \left( \frac{t}{z} - 1 \right)^{-\alpha}
	\mu (\d z)
	=
    t^{-\alpha} + \sum_{k=1}^\infty \binom{-\alpha}{k} t^{-\alpha - k} \int z^{k+\alpha} \mu ( \d z)
    =
    t^{-\alpha} + \O{t^{-1-\alpha}}
    ,
    \label{eq:nu_expansion_powerterm}
\end{equation}
where the sum is absolutely convergent for $\tau > c_+$.

As for the other term in Equation~\eqref{eq:nu_stationary_expansion}, we can make the change of variables $x = t/z$, use the fact that $\mu$ has a $C^1$ density to integrate by parts, and estimate $Q_\xi$ using Equation~\eqref{eq:Q_xi_bound} to obtain
\begin{equation}
	\begin{split}
	  &\int_{c_-}^{c_+} Q'_\xi \left( \frac{t}{z} \right) \mu (\d z)
		=
		t \int_{t/c_+}^{t/c_-} Q'_\xi (x) \mu \left( \frac{t}{x} \right) \frac{\d x}{x^2}
		\\ & \quad =
		t \int_{t/c_+}^{t/c_-} Q_\xi (x) \left[ \frac2{x^3}  \mu \left( \frac{t}{x} \right)
			+ \frac{2 t}{x^4} \mu' \left( \frac{t}{x} \right)
		\right] \d x
		= 
		\O{t^{-\alpha-\delta}}.
	\end{split}
\end{equation}
Substituting this and Equation~\eqref{eq:nu_expansion_powerterm} into Equation~\eqref{eq:nu_stationary_expansion}, we obtain Equation~\eqref{eq:Rnu}.

\end{proof}

\subsection{Asymptotics of $\omega_0$}
\label{sec:omega_asymptotics}

Defining a measure $\zeta$ by applying the change of variables $\tau = h_0(\sigma) = \sigma/(1+\sigma)$ to $\omega_0$, or in other words
$G_\zeta (\tfrac{\sigma}{1+\sigma}) = G_{\omega_0} (\sigma)$, then from the corresponding properties of $\omega_0$ it is clear that $\zeta$ is absolutely continuous and supported within $[0,c_+/(1+c_+)$.
 
Letting
\begin{equation}\label{eq:Z_def}
{\mathtt Z}(u) := \int \tau^u \zeta(\dd \tau),
\end{equation}
we will show the following counterpart of Lemma~\ref{lem:X_mellin}:

\begin{lemma}
\label{lem:Z_mellin_prelim}
There is a $U > 0$ such that ${\mathtt Z}(u)$ is analytic for $\Re u > U$, and has an analytic continuation for $0 \le \Re u \le U $ apart from the points where ${\mathtt M}(u) = 1$ with $u \neq 0$. 
\end{lemma} 
\begin{proof}
The control on the growth of $\omega_0$ at the origin demonstrated in Theorem~\ref{th:omega0}
there is some $U > 0$ such that for all $u \in \mathbb{C}$ with  $\Re u > U$ the integral in the definition~\eqref{eq:Z_def} of ${\mathtt Z}$ is absolutely convergent, and therefore defines an analytic function of $u$.

Equation~\eqref{eq:S_def_CDF} with $f(x) = x^u$ implies that $\zeta$ satisfies
\begin{equation}
\int \tau^u  \zeta(\d \tau) = \int \int \left( \frac{z \tau}{1 + z \tau}\right)^u \mu (\d z) \zeta (\d \tau).
\label{eq:zeta_mellin}
\end{equation}
Then  using the identity
\begin{equation}
\left( \frac{z \tau}{1 + z \tau}\right)^u 
=
\frac{1}{\Gamma(u)} \frac{1}{2 \pi i} 
\int_{w_0 - i \infty}^{w_0 + i \infty} (z\tau)^{w} \Gamma(u-w) \Gamma(w) \d w
\end{equation}
(obtained from the formula for the Mellin transform of the Beta function, and valid for $0 < w_0 < \Re u$), \eqref{eq:zeta_mellin} can be rewritten 
\begin{equation}
{\mathtt Z}(u) 
= 
\frac{1}{\Gamma(u)} \frac{1}{2 \pi i} 
\int_{w_0 - i \infty}^{w_0 + i \infty} \Gamma(u-w) \Gamma(w) {\mathtt Z}(w) {\mathtt M}(w) \d w.
\label{eq:zeta_prerecursion}
\end{equation}
Displacing the contour of integration to the right across the pole of $\Gamma(u-w)$ at $u=w$ and rearranging, we obtain the counterpart of \eqref{eq:X_recursion},
\begin{equation}
{\mathtt Z}(u) 
= 
\frac{1}{\Gamma(u) [1 - {\mathtt M}(u)]} \frac{1}{2 \pi i} \int_{w_0 - i \infty}^{w_0 + i \infty} \Gamma(u-w) \Gamma(w) {\mathtt Z}(w) {\mathtt M}(w) \d w,
\label{eq:Z_recursion}
\end{equation}
for $U < \Re u < w_0 <  \Re u + 1$.  The integral on the right-hand side can be analytically continued in $u$ so long as the contour of integration is displaced to maintain the condition $\Re u < w_0 <  \Re u + 1$ and $w_0 > 0$, allowing the whole expression to be analytically continued as well, apart from the zeros of $1 - {\mathtt M}(u)$.  Note that although $1-{\mathtt M}(0) = 0$, the associated pole is removable, thanks to the factor of $\Gamma(u)$.
\end{proof}

The fact that we have little control over the value of $U$ will make using this result inconvenient when we attempt to repeat the analysis of Section~\ref{sec:nu_asymptotics}, but we can refine the result as follows:

\begin{lemma}
The integral
\begin{equation}
\int \sigma^u \zeta ( \d \sigma)
\end{equation}
is absolutely convergent whenever $\Re u > \alpha$.
\label{lem:Z_strip}
\end{lemma}
\begin{proof}
Let
\begin{equation}
A := \inf \left\{ u \in \mathbb{R} \middle| \int \sigma^u \zeta ( \d \sigma)  < \infty \right\}.
\end{equation}
Suppose $A > \alpha$.  Then from Lemma~\ref{lem:Z_mellin_prelim}, $\alpha$ is the only pole of ${\mathtt Z}$ on the positive real axis, so ${\mathtt Z}$ is analytic at $A$ and there is a $u > A$ such that the Taylor series of ${\mathtt Z}$ at $u$ converges on an open disk containing $A$.
The Taylor series of ${\mathtt Z}$ at $u$ is 
\begin{equation}
\sum_{n=0}^\infty \frac{(u-v)^n}{n!} \int \sigma^u (- \log \sigma)^n \zeta( \d \sigma).
\end{equation}
Since $\log \sigma < 0$ on the support of $\zeta$, we can exchange the sum and integral whenever the Taylor series is absolutely convergent.  In particular, there is some $v < A$ for which this is true, and we have
\begin{equation}
\int \sigma^v \zeta ( \d \sigma)
=
\int \sum_{n=0}^\infty \frac{(u-v)^n}{n!} (- \log \sigma)^n \sigma^u \zeta ( \d \sigma)
< \infty,
\end{equation}
contradicting the definition of $A$.

This proves convergence on the real line; to extend to complex numbers, we simply note that
\begin{equation}
\int \left| \sigma^u \right| \zeta ( \d \sigma)
= 
\int \sigma^{\Re u} \zeta ( \d \sigma).
\end{equation}
\end{proof}

We can confirm that $\alpha$ actually is a pole of ${\mathtt Z}$ in the same way as we did for ${\mathtt X}$ in the previous section, using \eqref{eq:zeta_mellin} to obtain
\begin{equation}
\int \sigma^u \zeta(\d u)
\ge 
{\mathtt M}(u) \int \sigma^u \zeta (\d u),
\end{equation}
so that since ${\mathtt M}(u) < 1$ for $u \in (0,1)$ the integral must diverge there.

The relationship of ${\mathtt Z}$ to $G_\zeta$ is slightly simpler than what we saw in the previous sections: for $\Re u > U$ we can integrate \eqref{eq:Z_def} by parts to obtain
\begin{equation}
{\mathtt Z}(u) = u \int_0^\infty \sigma^{u-1} G_\zeta(\sigma) \d \sigma.
\label{eq:zeta_forward_mellin}
\end{equation}
Repeating the proof of Lemma~\ref{lem:Z_strip} we see that the integral on the right hand side is absolutely convergent for all $\Re u > \alpha$ and this expression holds everywhere on that half-plane by analytic continuation.
\begin{remark}\label{rem:omega_at_origin}
In particular, since $\alpha < 1$,
\begin{equation}
\int_0^\infty G_\zeta(\sigma) \d \sigma < \infty,
\end{equation}
and since $G_\zeta$ is a nonnegative, nonincreasing function this implies $G_\zeta (\sigma) = o(1/\sigma)$ (and therefore also $G_{\omega_0} (\sigma) = o(1/\sigma)$) for $\sigma \searrow 0$.
\end{remark}

Noting that $G_\zeta$ is continuous and monotone, we can apply the inverse Mellin formula to obtain 
\begin{equation}\label{eq:zeta_inverse_Mellin}
G_\zeta (\sigma) = \frac{1}{2 \pi i	} \int_{u_0 - i \infty}^{u_0 + i \infty} \frac{1}{u} \sigma^{-u} {\mathtt Z}(u) \d u
\end{equation}
for all $\sigma \ge 0$ where $G_\zeta$ is continuous and all $u_0 > \alpha$.

In order to displace the contour of integration in \eqref{eq:zeta_inverse_Mellin} as we did with \eqref{eq:G_inverse_Mellin}, we again need a little control over the growth of ${\mathtt Z}$ for large imaginary arguments.  This can be obtained in nearly the same way as was done for ${\mathtt X}$ in Lemma~\ref{lem:X_growth}:
\begin{lemma}
For $\Im u$ large with $\Re u = u_0 > \alpha - 1$ fixed,
\begin{equation}
|{\mathtt Z}(u)| = \O{|\Im u |^{1/2}}.
\end{equation}
\end{lemma}
\begin{proof}
Fixing some $w_0 \in (\alpha, u_0+1)$, \eqref{eq:Z_recursion} implies
\begin{equation}
|{\mathtt Z}(u)| 
\le 
\frac{{\mathtt M}(w_0)}{|1- {\mathtt M}(u)|} 
\frac{{\mathtt Z}(w_0)}{|\Gamma(u)|}
\frac{1}{2 \pi}
\int_{-\infty}^\infty \left|
\Gamma ( w_0 + ix) \Gamma(u - w_0 - i x )
\right| \d x,
\end{equation}
where we note that Lemma~\ref{lem:Z_strip} implies $|{\mathtt Z}(w_0+ix)| <  {\mathtt Z}(w_0) < \infty$. 
Then using the estimates in the proof of Lemma~\ref{lem:X_growth} with the signs of $u$, $w_0$ and $x$ reversed we arrive at the same estimate.
\end{proof}

Then we can displace the contour in \eqref{eq:zeta_inverse_Mellin} to obtain
\begin{equation}
G_\zeta (\sigma) = {\mathtt Z}(0) +  \frac{\res_{\mathtt Z}(\alpha)}{\alpha} \sigma^{-\alpha} + \frac{1}{2 \pi i	} \int_{u_0 - i \infty}^{u_0 + i \infty} \frac{1}{u} \sigma^{-u} {\mathtt Z}(u) \d u,
\label{eq:zeta_expansion}
\end{equation}
valid for some $u_0 < 0$, since Lemma~\ref{lem:Z_strip} implies that $\alpha$ is the only pole of ${\mathtt Z}$ in the right half plane.

As before, the integral on the right hand side of \eqref{eq:zeta_expansion} is not absolutely convergent, but is equal to the derivative of a function given by an absolutely convergent integral, in this case
\begin{equation}
Q_\zeta(\sigma) =  \frac{1}{2 \pi i} \int_{u_0 - i \infty}^{u_0 + i \infty} \frac{1}{u (1 - u)} \sigma^{1-u} {\mathtt Z}(u) \d u.
\end{equation}
This expression is manifestly is $O(\sigma^{1-u_0})$ (in particular $o(\sigma)$, since $u_0 < 0$), and we can use this to obtain the counterpart of Lemma~\ref{lem:nu_asymptotics}:
\medskip

\begin{lemma} \label{lem:omega_asymptotics}
There is a function $R_\omega:\mathbb{R}^+ \to \mathbb{R}$ and a constant $C_\omega \neq 0$ such that 
\begin{equation}
G_{\omega_0} (\sigma) = C_\omega \left[\sigma^{-\alpha} + R_\omega (\sigma) \right]
\label{eq:Gom_expansion}
\end{equation}
and 
\begin{equation}
	R_\omega (\sigma) =  \O{1}
	\label{eq:Rom_order}
\end{equation}
as $\sigma \searrow 0$.
\end{lemma}
The proof is the same as Lemma~\ref{lem:nu_asymptotics}, apart from a few details.
\begin{proof}
	Letting 
	$C_\omega := \res_{\mathtt Z}(\alpha)/\alpha$ and 
	$R_\omega(\sigma) := C_{\omega}^{-1} \left[ {\mathtt Z} (0) + Q'_\zeta(\sigma/(1-\sigma))
	\right]$ and recalling that $G_\zeta(\tfrac{\sigma}{1+\sigma}) = G_{\omega_0}(\sigma)$,
	we obtain the expansion~\eqref{eq:Gom_expansion} from Equation~\eqref{eq:zeta_expansion}.  Then using writing out the stationarity condition $\omega_0 = S_0 \omega_0$ as in~\eqref{eq:S_def_CDF2}, we have
	\begin{equation}
		\sigma^{-\alpha} + R_\omega (\sigma)
		=
		\int_\sigma^\infty \left[ \left( \frac{z-\sigma}{\sigma} \right)^\alpha + 
			Q'_\zeta \left( \frac{\sigma}{z} \right) 
		\right]\mu ( \d z) 
		+ {\mathtt Z} (0) G_\mu (\sigma)
		\label{eq:om_stationary_expansion}
	\end{equation}
	after cancelling a factor of $C_\omega$ and noting that $h_0^{-1}(y) = y/(1-y)$.
	By the generalized binomial theorem,
	\begin{equation}
		\left( \frac{z-\sigma}{\sigma} \right)^\alpha
		=
		\frac{z^\alpha}{\sigma^\alpha} \sum_{k=0}^{\infty} \binom{\alpha}{k} \left( -\frac{\sigma}{z} \right)^{-\alpha},
	\end{equation}
	where the sum is absolutely convergent for $\sigma < z$; then we have
	\begin{equation}
		\begin{split}
			\int_\sigma^\infty \left( \frac{z-\sigma}{\sigma} \right)^\alpha \mu (\d z)
			& = 
			\sigma^{-\alpha} 
			+ \sum_{k=1}^\infty (-1)^k \binom{\alpha}{k} \sigma^{k - \alpha} \int_\sigma^\infty z^{\alpha-k} \mu (\d z)
			\\ & =
			\sigma^{-\alpha} 
			+ \O{\sigma^{1-\alpha}}
		\end{split}
		\label{eq:om_expansion_powerterm}
	\end{equation}
	for $\sigma \searrow 0$.

	As for the other integral in~\eqref{eq:om_stationary_expansion}, for $\sigma < c_-$ we can make the change of variables $x = \sigma/z$, use the fact that $\mu$ has $C^1$ density to integrate by parts, and use the observation that $Q_\zeta(\sigma) = o(\sigma)$ to obtain
	\begin{equation}
		\begin{split}
			\int_\sigma^\infty Q_\zeta'\left( \frac{\sigma}{z} \right) \mu ( \d z)
			&=
			\sigma \int_{\sigma/c_+}^{\sigma/c_-} Q'_\zeta (x) \mu\left( \frac\sigma{x} \right) \frac{\d x}{x^2}
			\\ &= \sigma \int_{\sigma/c_+}^{\sigma/c_-} Q_\zeta (x)
			\left[ 
				\frac2{x^3}\mu\left( \frac\sigma{x} \right) 
				+\frac\sigma{x^4}\mu'\left( \frac\sigma{x} \right) 
			\right] \d x
			= o(1)
		\end{split}
	\end{equation}
	as $\sigma \searrow 0$.  Inserting this and \eqref{eq:om_expansion_powerterm} into~\eqref{eq:om_stationary_expansion}, and noting that $G_\mu(\sigma) \to 1$ as $\sigma \to 0$, we obtain Equation~\eqref{eq:Rom_order}.
\end{proof}

\section{An approximately stationary point} \label{sec:norm_bounds}

This section is devoted to the proof of Theorem~\ref{th:main}. Following the strategy outlined 
in Section~\ref{sec:approach} we introduce a measure $\gamma_\eps$ which is changed only slightly (as measured by $\vertiii{\cdot}_\beta$) by the action of $T_\eps$.  Applying Lemma~\ref{th:mainlem}, we see that this implies that $\gamma_\eps$ is close to the stationary measure $\nu_\eps$ in a way which allows us to use it to estimate the Lyapunov exponent $\mathcal{L}(\eps)$.

\medskip
\noindent
{\it Proof of Theorem~\ref{th:main}}.
For each $\eps$, in view of introducing the probability $\gamma_\gep$ define a measure $\gex$ by 
\begin{equation}
\label{eq:gex_def}
G_\gex ( x) = \left\{
\begin{array}{ll}
a(\eps) G_{\omega_0} (\eps^2 x), & x \ge \frac{1}{\eps}\\
a(\eps) G_{\omega_0} (\eps  ) + G_{\nu_0} (x) - G_{\nu_0} (1/\eps) , & x < \frac{1}{\eps}
\end{array}
\right.
\end{equation}
or equivalently
\begin{equation}
\label{eq:gex_def2}
F_\gex (x) = \left\{
\begin{array}{ll}
F_{\nu_0} (x), & x < \frac{1}{\eps}\\
F_{\nu_0} ( 1/\eps) + a(\eps) \left[ G_{\omega_0} (\eps ) - G_{\omega_0} (\eps^2 x ) \right], & x \ge \frac{1}{\eps}
\end{array}
\right. ,
\end{equation}
where $a(\eps) := (C_\nu / C_\omega ) \eps^{2 \alpha}$, so that (see 
\eqref{eq:Gnu_expansion} and \eqref{eq:Gom_expansion})
\begin{equation}
G_{\nu_0} (t) - a(\eps) G_{\omega_0} (\eps^2 t) = 
C_\nu
\left[
R_\nu (t) - \eps^{2 \alpha} R_\omega (\eps^2 t)
\right]. \label{eq:G_cancellation}
\end{equation}
In Section~\ref{sec:ghat} below, we will show that
\begin{equation}
\vertiii{T_\eps \hat{\gamma}_\eps - \hat{\gamma}_\eps}_\beta
=
\Oeps{2\alpha - \beta}
+\Oeps{\alpha+\delta - \beta}.
\label{eq:preview}
\end{equation}
This will be done by using the definition of $\gex$ as a piecewise expression in terms of the stationary measures $\nu_0 = T_0 \nu_0$ and $\omega_0 = S_0 \omega_0$ to write out the above distance as an integral of terms with approximate calculations due to the presence of differences either of the form $T_\eps - T_0$ or $S_\eps - S_0$, or of the form of the left hand side of Equation~\eqref{eq:G_cancellation}.

We cannot apply Lemma~\ref{th:mainlem} immediately, because $\hat{\gamma}_\eps$ is not a probability measure.  However 
using Lemma~\ref{lem:nu_asymptotics} 
and Lemma~\ref{lem:omega_asymptotics},
\begin{equation}\label{eq:gamma_normalization}
	G_\gex (0)
=
a(\eps) G_{\omega_0} (\eps ) + 1 - G_{\nu_0}(1/\eps)
=
1+\Oeps{\alpha},
\end{equation}
and thus if we define a probability measure
\begin{equation}
\gamma_\eps = \frac{\hat{\gamma}_\eps }{G_\gex(0)}
\end{equation}
we have
\begin{equation}
\label{eq:formainlem}
\vertiii{T_\eps \gamma_\eps - \gamma_\eps}_\beta
\, =\, 
\left[
1+ \Oeps{\alpha}
\right]
\vertiii{T_\eps \hat{\gamma}_\eps - \hat{\gamma}_\eps}_\beta
=
\Oeps{2\alpha - \beta}
+\Oeps{\alpha+\delta - \beta}\, .
\end{equation}
Now we write
\begin{multline}
L_\eps[\gex] \, =\, 
\int \log(1 + \eps^2 s) \gex (\d s)
\\
=\,  \int \log(1 + \eps^2 s) \ind_{[0,1/\eps)}(s) \nu_0(\d s)
+
a(\eps) \int \log(1+\sigma) \ind_{[\eps ,\infty)}(\sigma) \omega_0(\d s)
\\
=\,
 a(\eps)
\int \log (1+\sigma) \omega_0(\d \sigma) 
+ \int_0^{1/\gep} \log\left(1+ \gep^2s \right) \nu_0(\dd s) 
\\ - a(\gep) \int_0^{\gep } \log(1+ \gs) \go_0(\dd \gs)\, .
\end{multline}
The last two terms can be bounded as follows:
\begin{equation}
\begin{split} 
0 \le &
\int_0^{1/\gep} \log\left(1+ \gep^2s \right) \nu_0(\dd s) 
\le 
\eps^2 \int_0^{1/{\eps}} s \ \nu_0 (\d s)
\\
& =
- \eps^2 \frac{1}{\eps} G_{\nu_0}\left( 1/{\eps} \right)
+ \eps^2 \int_0^{1/{\eps}} G_{\nu_0}(s) \d s
= \Oeps{\alpha + 1},
\end{split}
\end{equation}
where we have used Lemma~\ref{lem:nu_asymptotics} for the last estimate, and similarly
\begin{multline}
0 \le 
\int_0^{\gep } \log(1+ \gs) \go_0(\dd \gs)
\le 
\int_0^{\gep } \gs \go_0(\dd \gs)
=
- \eps  G_{\omega_0} (\eps )
+ \int_0^{\gep } G_{\go_0} (\gs) \dd \gs
\\ =
\Oeps{1-\alpha}
\end{multline}
using Remark~\ref{rem:omega_at_origin} to conclude that the other boundary term is zero and using Lemma~\ref{lem:omega_asymptotics} to estimate the final integral.  
Combining the last three equations and taking into account the correction from \eqref{eq:gamma_normalization}, we obtain
\begin{equation}
\label{eq:formainlem-2}
L_\eps(\gamma_\eps) \, =\,  \left( 1 + \Oeps{\alpha} \right)
a(\eps) \int \log (1+\sigma) \omega_0(\d \sigma) + \Oeps{1+\alpha} \, .
\end{equation}
The integral appearing here is finite: 
\begin{equation}
\begin{split}
\int \log (1+\sigma) \omega_0(\d \sigma) 
= - \int \log (1-\tau)
\zeta (\d \tau)
<
\frac{(1+c_+) \log(1+c_+)}{c_+}
\int \tau \zeta (\d \tau) 
\\ = \frac{(1+c_+) \log(1+c_+)}{c_+} {\mathtt Z}(1) < \infty
\end{split}
\end{equation}
for $\zeta$ and ${\mathtt Z}$ defined in Section~\ref{sec:omega_asymptotics} (noting in particular that the support of $\zeta$ is contained in $[0,\frac{c_+}{1+c_+}]$).  We can then apply  Lemma~\ref{th:mainlem}, using \eqref{eq:formainlem} and \eqref{eq:formainlem-2}, and 
the proof of Theorem~\ref{th:main} is complete, modulo of course establishing \eqref{eq:preview} to which all the rest of the section is devoted, with
\begin{equation}
\label{eq:Cmu}
C_\mu := 
\frac{C_\nu}{C_\omega} 
\int \log (1 + \sigma) \, \omega_0 (\d \sigma).
\end{equation} 
\qed

\medskip

\begin{remark}
\label{rem:uni2}
Any change in the normalization of $\omega_0$ is cancelled by a change in $C_\omega$, so this definition of $C_\mu$ is indeed independent of this normalization.
\end{remark}

\subsection{Quasi-stationarity estimates for $\hat{\gamma}_\eps$: proof of \eqref{eq:preview}}\label{sec:ghat}

Writing out the definitions of $\vertiii{\cdot}_\beta$ and $\gex$ and using the stationarity properties $\nu_0 = T_0 \nu_0$ and $\omega_0 = S_0 \omega_0$ we obtain
\begin{multline}
\vertiii{T_\eps \hat{\gamma}_\eps - \hat{\gamma}_\eps}_\beta
=
\int_{c_-}^{c_+\eps^{-2}} 
\tau^{\beta - 1} 
\left|
 \left[
F_{T_\eps \gex}(\tau) - F_\gex (\tau)
\right]
\right|
\d \tau
\\
=
\int_{c_-}^{c_+\eps^{-2}}
\tau^{\beta - 1}
\left|
\left\{
\ind_{[1/\eps,\infty]}(\tau)
\left[
a(\eps) G_{S_0 \omega_0} (\eps^2 \tau ) - G_{T_\eps \gex} (\tau)
\right] \right. \right.
\\
+ \left. \left. \ind_{[0,1/\eps)} (\tau)
\left[
F_{T_\eps \gex} (\tau) - F_{T_0\nu_0} (\tau)
\right]
\vphantom{\eps^2}
\right\} 
\vphantom{\int}
\right| \d \tau\, ,
\label{eq:presplit1}
\end{multline}
and by
 \eqref{eq:alt-G} and \eqref{eq:T_def_CDF}
we can rewrite \eqref{eq:presplit1} as
\begin{multline}
\vertiii{T_\eps \hat{\gamma}_\eps - \hat{\gamma}_\eps}_\beta\, =
\\
\int_{c_-}^{c_+\eps^{-2}}
\tau^{\beta - 1}
\left|
\ind_{[1/\eps,\infty]}(\tau)
 \int_{\eps^2\tau}^\infty
\left[ 
a(\eps) G_{\omega_0} \left( h_0^{-1} \left(\frac{\eps^2 \tau}{t} \right) \right)
-
G_\gex \left( g_\eps^{-1} \left(\frac{\tau}{t} \right) \right)
\right] \right.
 \mu (\d t) 
\\ +
 \ind_{[0,1/\eps)}(\tau) \left\{
\int
\left[ 
F_\gex \left( g_\eps^{-1} \left(\frac{ \tau}{t} \right) \right)
-
F_{\nu_0} \left( g_0^{-1} \left(\frac{ \tau}{t} \right) \right)
\right] \mu (\d t) \right. 
\\ + \left. \left.[F_{\gex}(\infty) - 1] F_\mu(\eps^2 \tau)\vphantom{\int}
\right\} 
\right| 
\d \tau.
\end{multline}
We can simplify this somewhat by restricting to $\eps \le c_- $; 
then $\eps^2 \tau \le c_-$ whenever $\tau < 1/\eps$, and therefore $ \ind_{[0,1/\eps)}(\tau) F_\mu(\eps^2 \tau)=0$.
Then writing out the definition of $\gex$ in \eqref{eq:gex_def} and  \eqref{eq:gex_def2}, 
we have
\begin{equation}
\begin{split}
& \vertiii{T_\eps \hat{\gamma}_\eps - \hat{\gamma}_\eps}_\beta\, =
\\ & \quad
\int_{c_-}^{c_+\eps^{-2}}
\tau^{\beta - 1}
\left|
	\int 
	\bigg\{
		\ind_{[1/\eps,\infty)}(\tau) 
			\ind_{[\eps^2 \tau,\eps \tau]}(t)
			\right.
\\ 
	& \qquad  \qquad\times
	\left[
		a(\eps)G_{\omega_0} \left( h_0^{-1} \left(\frac{\eps^2 \tau}{t} \right) \right) 
		- a(\eps) G_{\omega_0} \left(\eps^2 g_\eps^{-1} \left(\frac{\tau}{t} \right) \right)
	\right]
\\ & \qquad
+ 
\ind_{[1/\eps,\infty)}(\tau) \ind_{(\eps \tau,\infty)}(t)
\\
& \qquad  \qquad \times
\left[
a(\eps) G_{\omega_0} \left( h_0^{-1} \left(\frac{\eps^2 \tau}{t} \right) \right) 
-
a(\eps) G_{\omega_0} (\eps )
-
G_{\nu_0} \left( g_\eps^{-1} \left(\frac{ \tau}{t} \right) \right)
+
G_{\nu_0} \left( {\eps}^{-1} \right)
\right]
\\ & \qquad
+ 
\ind_{[0,1/\eps)}(\tau) \ind_{[0,\eps \tau]} (t)
\\
& \qquad  \qquad\times
\left[
F_{\nu_0}  \left( \frac{1}{\eps} \right)
+
a(\eps) G_{\omega_0} (\eps )
-
a(\eps) G_{\omega_0} \left( \eps^2 g_\eps^{-1} \left( \frac{\tau}{t}\right) \right)
-
F_{\nu_0} \left( g_0^{-1} \left(\frac{\tau}{t}\right) \right)
\right]
\\ & \qquad
+ \left. 
\ind_{[0,1/\eps)}(\tau) \ind_{(\eps \tau,\infty)} (t)
\left[
F_{\nu_0} \left( g_\eps^{-1} \left(\frac{\tau}{t}\right) \right)
-
F_{\nu_0} \left( g_0^{-1} \left(\frac{\tau}{t}\right) \right)
\right]
\vphantom{G_{\omega_0} \left( h_0^{-1} \left(\frac{\eps^2 \tau}{t} \right) \right)}
\bigg\} \mu(\d t)
\vphantom{\int} \right|
\d \tau,
\end{split}
\end{equation}
where we have used the observation that $g_\eps^{-1}(\tau/t) \ge 1/\eps$ is equivalent to $t \le \eps \tau$ since $g_\eps(1/\eps) = 1/\eps$ to simplify the indicator functions.

Rewriting the $F$s as $G$s, and using $\eps^2 g_\eps^{-1}(\tau / t) = h_\eps^{-1}(\eps^2 \tau/t)$ (which can be checked from \eqref{eq:g_inverse} and \eqref{eq:h_inverse}) in the first line 
\begin{equation}
\begin{split}
&
\vertiii{T_\eps \hat{\gamma}_\eps - \hat{\gamma}_\eps}_\beta =
\\ & 
\int
\tau^{\beta - 1}
\left|
	\int 
	\left\{
		\ind_{[1/\eps,\infty)}(\tau) 
			\ind_{[\eps^2 \tau,\eps \tau]}(t)
		\vphantom{\left( h_0^{-1} \left(\frac{\eps^2 \tau}{t} \right) \right)}
	\right.
\right.
\\ & \qquad \qquad \times
\left[
a(\eps)G_{\omega_0} \left( h_0^{-1} \left(\frac{\eps^2 \tau}{t} \right) \right) 
- a(\eps) G_{\omega_0} \left(h_\eps^{-1} \left(\frac{\eps^2 \tau}{t} \right) \right)
\right]
\\ & \quad
+ 
\ind_{[1/\eps,\infty)}(\tau) \ind_{(\eps \tau,\infty)}(t)
	\\
& \qquad  \qquad \times
\left[
a(\eps) G_{\omega_0} \left( h_0^{-1} \left(\frac{\eps^2 \tau}{t} \right) \right) 
-
a(\eps) G_{\omega_0} (\eps )
-
G_{\nu_0} \left( g_\eps^{-1} \left(\frac{ \tau}{t} \right) \right)
+
G_{\nu_0} \left( \frac{1}{\eps} \right)
\right]
\\ & \quad
+ 
\ind_{[0,1/\eps)}(\tau) \ind_{[0,\eps \tau]} (t)
\\
& \qquad  \qquad \times
\left[
G_{\nu_0} \left( g_0^{-1} \left(\frac{\tau}{t}\right) \right)
-
G_{\nu_0}  \left( \frac{1}{\eps} \right)
+
a(\eps) G_{\omega_0} (\eps )
-
a(\eps) G_{\omega_0} \left( \eps^2 g_\eps^{-1} \left( \frac{\tau}{t}\right) \right)
\right]
\\ & \quad
+ \left. \left.
\ind_{[0,1/\eps)}(\tau) \ind_{(\eps \tau,\infty)} (t)
\left[
G_{\nu_0} \left( g_0^{-1} \left(\frac{\tau}{t}\right) \right)
-
G_{\nu_0} \left( g_\eps^{-1} \left(\frac{\tau}{t}\right) \right)
\right]
\vphantom{\left( h_0^{-1} \left(\frac{\eps^2 \tau}{t} \right) \right)}
\right\}  \mu (\d t)
\right| 
\d \tau.
\end{split}
\end{equation}
We can then use the Triangle inequality to split up the integrals into four parts, obtaining

\begin{equation}
\begin{split}
&\vertiii{T_\eps \hat{\gamma}_\eps - \hat{\gamma}_\eps}_\beta
\le
\\
& \quad 
	a(\eps) \int_{1/\eps}^{c_+\eps^{-2}}
	\tau^{\beta - 1}
	\int_{\eps^2 \tau}^{\eps\tau}
	\left[ 
		G_{\omega_0} \left(h_\eps^{-1} \left(\frac{\eps^2 \tau}{t} \right) \right)
		- G_{\omega_0} \left( h_0^{-1} \left(\frac{\eps^2 \tau}{t} \right) \right) 
	\right] 
	 \mu (\d t) \d \tau
\\ & \quad
+
	\int_{c_-}^{1/\eps}
	\tau^{\beta - 1}
	\int_{\eps \tau}^\infty
	\left[
	G_{\nu_0} \left( g_0^{-1} \left(\frac{\tau}{t}\right) \right)
	-
	G_{\nu_0} \left( g_\eps^{-1} \left(\frac{\tau}{t}\right) \right)
	\right] 
	\mu (\d t) \d \tau
\\ & \quad
+
	\int_{1/\eps}^{c_+\eps^{-2}}
	\tau^{\beta - 1}
	\left| \int_{\eps\tau}^\infty \right.
\\ & \qquad \left. \times 
	\left[
	a(\eps) G_{\omega_0} \left( h_0^{-1} \left(\frac{\eps^2 \tau}{t} \right) \right) 
	-
	a(\eps) G_{\omega_0} (\eps )
	-
	G_{\nu_0} \left( g_\eps^{-1} \left(\frac{ \tau}{t} \right) \right)
	+
	G_{\nu_0} \left( \frac{1}{\eps} \right)
	\right]\right.
\\ & \qquad \qquad \times \left. \vphantom{\int} 	
	\mu (\d t) 
	\right| \d \tau
\\ & \quad
+
	\int_{c_-}^{1/\eps}
	\tau^{\beta - 1}
	\left| \int_0^{\eps \tau}
		\right. 
\\ & \left.  \qquad \times 
	\left[
	G_{\nu_0} \left( g_0^{-1} \left(\frac{\tau}{t}\right) \right)
	-
	G_{\nu_0}  \left( \frac{1}{\eps} \right)
	+
	a(\eps) G_{\omega_0} (\eps )
	-
	a(\eps) G_{\omega_0} \left( \eps^2 g_\eps^{-1} \left( \frac{\tau}{t}\right) \right)
	\right] \right.
\\ & \qquad \qquad \times 
	\left. \vphantom{\int}	
	\mu (\d t) 
	\right| \d \tau
\end{split}
\label{eq:big_split}
\end{equation}
where the fact that $h_\eps^{-1} (y) \le h_0^{-1}(y)$ and 
$g_\eps^{-1} (y) \ge g_0^{-1}(y)$ for the relevant $y$ implies that the first two integrands are non negative.

Let us now examine the four terms in \eqref{eq:big_split} in turn.  The first one can be rewritten (mainly using the fact that the integrand is nonnegative)
\begin{equation}
\begin{split}
&
	a(\eps) \int_{1/\eps}^{c_+\eps^{-2}}
	\tau^{\beta - 1}
	\int_{\eps^2 \tau}^{\eps\tau}
	\left[ 
		G_{\omega_0} \left(h_\eps^{-1} \left(\frac{\eps^2 \tau}{t} \right) \right)
		- G_{\omega_0} \left( h_0^{-1} \left(\frac{\eps^2 \tau}{t} \right) \right) 
	\right] 
\\
& \le 
a(\eps) \int \int_{1/{\eps} }^\infty 
\tau^{\beta - 1}
\left[ 
G_{\omega_0} \left(h_\eps^{-1} \left(\frac{\eps^2 \tau}{t} \right) \right)
- G_{\omega_0} \left( h_0^{-1} \left(\frac{\eps^2 \tau}{t} \right) \right) 
\right]  \d \tau \ \mu (\d t)
\\
& =
a(\eps) \int \int \int_{1/\eps }^\infty \left[ 
	\ind_{[0,\frac{t}{\eps^2}h_\eps(\sigma)]}(\tau)
-
\ind_{[0,\frac{t}{\eps^2}h_0(\sigma)]}(\tau)
\right] \tau^{\beta-1} \d \tau \ \omega_0(\d \sigma) \  \mu (\d t)
\\
& =
a(\eps) \int \int \left[
	\ind_{(h_\eps^{-1}(\eps/t),\infty)} (\sigma)
	\int_{1/\eps}^{t\eps^{-2} h_\eps(\sigma)} \tau^{\beta - 1} \d \tau  
\right.
\\ & \qquad \qquad \qquad \qquad -
\left.
	\ind_{(h_0^{-1}(\eps/t),\infty)} (\sigma)
	\int_{1/\eps}^{t\eps^{-2} h_0(\sigma)} \tau^{\beta - 1} \d \tau 
\right] 
\\ & \qquad \qquad \qquad
\times \omega_0(\d \sigma) \mu(\d t)
\\
& =
\left( \frac{1}{\eps}\right)^\beta
\frac{a(\eps)}{\beta} \int  
\left[
G_{\omega_0} \left( h_\eps^{-1} \left( \frac{\eps}{t} \right) \right)
 -
G_{\omega_0} \left( h_0^{-1} \left( \frac{\eps}{t} \right) \right)
\right] \mu (\d t)
\\ & \quad 
+
\frac{a(\eps) }{\beta \eps^{2\beta}} \int t^\beta
\int
\left[
\ind_{(h_\eps^{-1}(\eps/t),\infty)} (\sigma) \left(
 h_\eps(\sigma)
\right)^\beta
-
\ind_{(h_0^{-1}(\eps/t),\infty)} (\sigma) \left(
 h_0(\sigma)
\right)^\beta
\right] 
\\ & \qquad \qquad \qquad \times \omega_0(\d \sigma) \ 
 \mu (\d t) \ .
\end{split}
\label{eq:part1}
\end{equation}
As for the first integral on the right hand side, noting that
$$
h_0^{-1}\left(\frac{\eps}{t}\right) - h_\eps^{-1} \left(\frac{\eps}{t}\right)
=
\frac{t \eps^2}{t - \eps},
$$
it can be estimated using Lemma~\ref{lem:omega_asymptotics} as
\begin{equation}
\begin{split}
&\int  
\left[
G_{\omega_0} \left( h_\eps^{-1} \left( \frac{\eps}{t} \right) \right)
 -
G_{\omega_0} \left( h_0^{-1} \left( \frac{\eps}{t} \right) \right)
\right] \mu (\d t)
\\
&=
C_\omega \int \left[
\left(h_\eps^{-1} \left( \frac{\eps}{t} \right) \right)^{-\alpha}
-
\left(h_0^{-1} \left( \frac{\eps}{t} \right) \right)^{-\alpha}
+
R_\omega \left(h_\eps^{-1} \left( \frac{\eps}{t} \right)\right)
-
R_\omega \left(h_0^{-1} \left( \frac{\eps}{t} \right) \right)
\right] 
\\ & \qquad \qquad \times \mu (\d t)
\\
&\le 
C_\omega \int \left[
\alpha \frac{t\eps^2}{t - \eps} \left( \frac{t - \eps }{\eps } \right)^{\alpha+1}
+
R_\omega \left(h_\eps^{-1} \left( \frac{\eps}{t} \right)\right)
-
R_\omega \left(h_0^{-1} \left( \frac{\eps}{t} \right) \right)
\right] \mu (\d t)
= \O{1}
\end{split}
\label{eq:part1a}
\end{equation}
for $\eps \searrow 0$.  Multiplying by $a(\eps)/\eps^\beta$, we see that the first term in \eqref{eq:part1} is $O(\eps^{2 \alpha - \beta})$.

The second integral on the rightmost side of \eqref{eq:part1} can be written as
\begin{equation}
\begin{split}
&\int t^\beta
\int
\left[
\ind_{(h_\eps^{-1}(\eps/t),\infty)} (\sigma) \left(
 h_\eps(\sigma)
\right)^\beta
-
\ind_{(h_0^{-1}(\eps/t),\infty)} (\sigma) \left(
 h_0(\sigma)
\right)^\beta
\right] \omega_0(\d \sigma)
 \mu (\d t)
\\ &
=
\int t^\beta \int \left[
\ind_{(h_\eps^{-1}(\eps/t),h_0^{-1}(\eps/t)]}(\sigma)
h_\eps(\sigma)^\beta
+
\ind_{(h_0^{-1}(\eps/t),\infty)} (\sigma)
\left\{
h_\eps(\sigma)^\beta - h_0(\sigma)^\beta
\right\}
\right] 
\\ 
& \phantom{movemovemovemovemovemovemovemovemovemovemove}
\times\omega_0 (\d \sigma) \mu (\d t)
\\ &
\le 
c_+^\beta \int \left[h_\eps\left(h_0^{-1}\left(\frac{\eps}{t}\right) \right) \right]^\beta 
\left[
G_{\omega_0} \left(h_\eps^{-1}\left(\frac{\eps}{t}\right) \right) \mu (\d t)
-
G_{\omega_0} \left(h_0^{-1}\left(\frac{\eps}{t}\right) \right)\right] \mu (\d t)
\\ & \quad
+
\beta c_+^\beta \int \int \ind_{(h_0^{-1}(\eps/t),\infty)} (\sigma)
\left(
\frac{\eps^2}{1+\sigma}
\right)
\left(
\frac{\eps^2 + \sigma}{1+\sigma}
\right)^{\beta-1}
\omega_0 (\d \sigma) \mu (\d t)
\end{split}
\label{eq:part1b}
\end{equation}
The first integral in the final expression is similar to the one estimated in \eqref{eq:part1a}, apart from the presence of a factor of order $\eps^\beta$, and so the whole term is $O(\eps^\beta)$.
As for the second term, by integrating by parts and then applying Lemma~\ref{lem:omega_asymptotics} in the same fashion as \eqref{eq:part1a}, it can be bounded in the following way:
\begin{multline}
\int \int \ind_{(h_0^{-1}(\eps/t),\infty)} (\sigma)
\left(
\frac{\eps^2}{1+\sigma}
\right)
\left(
\frac{\eps^2 + \sigma}{1+\sigma}
\right)^{\beta-1}
\omega_0 ( \d\sigma)\  \mu (\d t)
\\
\le \eps^2 (c_+ +\eps^2)^{\beta-1} \int G_{\omega_0} \left(h_0^{-1}\left(\frac{\eps}{t}\right) \right) \mu (\d t)
=
O\left( \eps^{2-\alpha} \right)\, .
\end{multline}
Then the right hand side of \eqref{eq:part1b} is $O(\eps^\beta)$ (since $\beta < 1 < 2 - \alpha$), and so the second term on the right hand side of \eqref{eq:part1} is $O(\eps^{2\alpha - \beta})$ when the prefactor is included.  We already obtained an estimate of the same order for the first term, so we arrive at the estimate
\begin{multline}
	a(\eps)  \int_{1/\eps}^{c_+\eps^{-2}}
	\tau^{\beta - 1}
	\int
	 \ind_{[0,1/\eps)} \left( g_\eps^{-1} \left( \frac\tau{t} \right) \right) 
\\
 \qquad \times 
	\left[ 
		G_{\omega_0} \left(h_\eps^{-1} \left(\frac{\eps^2 \tau}{t} \right) \right)
		- G_{\omega_0} \left( h_0^{-1} \left(\frac{\eps^2 \tau}{t} \right) \right) 
	\right] 
	 \mu (\d t) \d \tau
\, 
\\  \quad =\,
\Oeps{2\alpha-\beta} 
\label{eq:bound1}
\end{multline}
for the first term on the right hand side of \eqref{eq:big_split}.

The next term is fairly similar:
\begin{equation}
\begin{split}
& 
	\int_{c_-}^{1/\eps}
	\tau^{\beta - 1}
	\int_{\eps \tau}^\infty
	\left[
	G_{\nu_0} \left( g_0^{-1} \left(\frac{\tau}{t}\right) \right)
	-
	G_{\nu_0} \left( g_\eps^{-1} \left(\frac{\tau}{t}\right) \right)
	\right] 
	\mu (\d t) \d \tau
\\ & \quad 
\le \int  \int_0^{1/\eps} 
\left[
F_{\nu_0} \left( g_\eps^{-1} \left(\frac{\tau}{t}\right) \right)
-
F_{\nu_0} \left( g_0^{-1} \left(\frac{\tau}{t}\right) \right)
\right] \tau^{\beta-1} \d \tau \, \mu (\d t)
\\ & \quad 
=
\int \int \left[
\ind_{[0,g_\eps^{-1}(1/t\eps)]}(s)
\int_{tg_\eps(s)}^{1/\eps} \tau^{\beta-1} \d \tau
\right.
\\ & \qquad \qquad \qquad 
\left.
-
\ind_{[0,g_0^{-1}(1/t\eps)]}(s)
\int_{tg_0(s)}^{1/\eps} \tau^{\beta-1} \d \tau
\right] 
\nu_0(\d s)\, \mu (\d t)
\\ &
=
\frac{1}{\beta}
\int \int
\left\{
	\ind_{[0,g_\eps^{-1}(1/t\eps)]}(s)
	\left[
		\eps^{-\beta}
		- \left( t g_\eps(s) \right)^\beta
	\right]
	\right.
	\\ & \qquad \qquad \qquad 
	\left.
	- 
	\ind_{[0,g_0^{-1}(1/t\eps)]}(s)
	\left[
		\eps^{-\beta}
		- \left( t g_0(s) \right)^\beta
	\right]
\right\}
\nu_0(\d s) \, \mu (\d t)
\\ &
	=
	\frac{1}{\beta \eps^\beta} 
	\int
	\left[
		G_{\nu_0} \left( g_0^{-1} 
		\left( \frac{1}{t \eps }\right)\right)
		-
		G_{\nu_0} \left( g_\eps^{-1} 
		\left( \frac{1}{t \eps }\right)\right)
	\right] \mu (\d t)
\\& \quad
	-
	\int
	t^\beta
	\left\{
		\int 
		\left[
			\ind_{[0,g_\eps^{-1}(1/t\eps)]}(s)
			\left( t g_\eps(s) \right)^\beta
			-
			\ind_{[0,g_0^{-1}(1/t\eps)]}(s)
			\left( t g_0(s) \right)^\beta
		\right]
		\nu_0(\d s)
	\right\} \mu (\d t)
\end{split}
\label{eq:part2}
\end{equation}
Noting that
$$
g_\eps^{-1}\left( \frac{1}{t\eps}\right) 
- g_0^{-1}\left( \frac{1}{t\eps}\right) 
= \frac{1-t\eps }{\eps(t-\eps)} - \frac{1-t\eps }{t\eps }
=
\frac{1-t\eps }{t(t-\eps)}
$$
and so
\begin{multline}
\left( g_0^{-1}\left( \frac{1}{t\eps}\right)\right)^{-\alpha}
-
\left( g_\eps ^{-1}\left( \frac{1}{t\eps}\right)\right)^{-\alpha} 
\le 
\alpha \frac{1-t\eps }{t(t-\eps)}
\left(
	\frac{t \eps }{1 - t \eps}
\right)^{\alpha+1}
\\
=
 \alpha 
\frac{t^\alpha \eps^{\alpha+1}}{(t-\eps)(1-t\eps)^\alpha}
,
\end{multline}
the first integral on the right hand side of \eqref{eq:part2} can be estimated as follows:
\begin{multline}
 \int \left[
G_{\nu_0} \left( g_0^{-1} \left( \frac{1}{t\eps} \right) \right) 
- G_{\nu_0} \left( g_\eps^{-1} \left( \frac{1}{t\eps} \right) \right) 
\right]  \mu (\d t)\, 
=\\
C_\nu
\int \bigg[
\left( g_0^{-1}\left( \frac{1}{t\eps}\right)\right)^{-\alpha}
-
\left( g_\eps ^{-1}\left( \frac{1}{t\eps}\right)\right)^{-\alpha}
+
R_\nu \left( g_0^{-1}\left( \frac{1}{t\eps}\right)\right)
-R_\nu \left( g_\eps ^{-1}\left( \frac{1}{t\eps}\right)\right)
\bigg]
\\ \times \mu (\d t)
\\ 
\le 
	 \alpha C_\nu \eps^{\alpha+1}
	\int
		\frac{t^\alpha}{(t-\eps )(1-t\eps)^\alpha}
		\mu (\d t) \phantom{movemovemovemovemove}
		\\
	+ C_\nu \int \left[
		R_\nu \left( g_0^{-1}\left( \frac{1}{t\eps}\right)\right)
		-R_\nu \left( g_\eps ^{-1}\left( \frac{1}{t\eps}\right)\right)
	\right] \mu (\d t)\,
= \, O\left( \eps^{\alpha+\delta} \right)\, ,
\label{eq:part2a}
\end{multline}
where we have used  Lemma~\ref{lem:nu_asymptotics} and observed that $g_\eps^{-1} (1/t\eps) = \Oeps{-1}$, $g_0^{-1} (1/t\eps) = \Oeps{-1}$ as $\eps \searrow 0$, uniformly for $t \in [c_-,c_+]$.
Consequently, the corresponding term in \eqref{eq:part2} is $\O{\eps^{\alpha+\delta-\beta}}.$

The inner part of the second integral on the rightmost side of \eqref{eq:part2} can be rewritten as
\begin{equation}
\begin{split}
	& \int 
	\left[
	\ind_{[0,g_0^{-1}(1/t\eps)]}(s)
		\left( t g_0(s) \right)^\beta
-
		\ind_{[0,g_\eps^{-1}(1/t\eps)]}(s)
		\left( t g_\eps(s) \right)^\beta
			\right]
	\nu_0(\d s)
\\& \quad
	= t^\beta \int
	\left\{
		\ind_{[0,g_\eps^{-1}(1/t\eps)]}(s)
		\left[ g_0(s)^\beta - g_\eps(s)^\beta \right]
		-
		\ind_{[g_0^{-1}(1/t\eps),g_\eps^{-1}(1/t\eps)]}(s)
		\left( g_\eps(s) \right)^\beta
	\right\}
\\& \qquad \qquad \qquad \times 
	\nu_0(\d s)
\\ & \quad
	\le 
		t^\beta\int
		\ind_{[0,g_\eps^{-1}(1/t\eps)]}(s)
		\left[ g_0(s)^\beta - g_\eps(s)^\beta \right]
		\nu_0(\d s)
\\ & \quad
	\le 
		\beta \eps^2 t^\beta \int 
		\ind_{[0,g_\eps^{-1}(1/t\eps)]}(s)
		s 
		\left(
			1+s
		\right)^\beta
		\nu_0(\d s)
\\ & \quad
	\le 
		\frac{\beta \eps^2 t^\gb }{(t-\gep)}
		\int (1+s)^\gb \nu_0(\dd s)
	\, \le \, 4
	\beta \eps t^{\gb-1} 
		\int \left(1+s^\gb\right) \nu_0(\dd s) \, = \, 
		\Oeps{},
\end{split}
\label{eq:part2b}
\end{equation}
where the second to last inequality uses the observation that $g_\gep^{-1}(1/(t\gep))\ge s$ implies $s\le 1/(\gep(t-\gep))$,
and the final estimate uses $\int s^\gb \nu_0(\dd s) < \infty$
(cf. \eqref{eq:technu0} and line right after).
With this we see that the right hand side of \eqref{eq:part2} is $\Oeps{\alpha+\delta-\beta}$ (noting $\alpha+\delta - \beta < \alpha+\delta < 1$).

We rewrite the third term in \eqref{eq:big_split} using Lemmas~\ref{lem:nu_asymptotics} and~\ref{lem:omega_asymptotics} and the triangle inequality,
\begin{equation}
\begin{split}
	&
	\int_{1/\eps}^{c_+\eps^{-2}}
	\tau^{\beta - 1}
	\left| \int_{\eps\tau}^\infty \right.
\\ & \qquad \left. \times 
	\left[
	a(\eps) G_{\omega_0} \left( h_0^{-1} \left(\frac{\eps^2 \tau}{t} \right) \right) 
	-
	a(\eps) G_{\omega_0} (\eps )
	-
	G_{\nu_0} \left( g_\eps^{-1} \left(\frac{ \tau}{t} \right) \right)
	+
	G_{\nu_0} \left( \frac{1}{\eps} \right)
	\right]\right.
\\ & \qquad \qquad \times \left. \vphantom{\int} 	
	\mu (\d t) 
	\right| \d \tau
\\ &\quad
\le
	C_\nu
	\int_{1/\eps}^\infty \tau^{\beta-1}
	\left|
		\int_{\eps\tau}^\infty  
		\left[
			\left(
				\frac{1}{\eps^2}
				h_0^{-1}\left( \frac{\eps^2 \tau}{t}\right)
			\right)^{-\alpha}
			-
			\left(
				g_\eps^{-1}\left( \frac{\tau}{t}\right)
			\right)^{-\alpha}
		\right] 
		\mu (\d t)
	\right|
	   \d \tau
\\ & \quad  \quad
+
	\ C_\nu \eps^{2\alpha}
	\int_{1/\eps}^\infty \tau^{\beta-1}
	\left|
		\int_{\eps\tau}^\infty  
		R_\omega \left( h_0^{-1} \left(\frac{\eps^2 \tau}{t} \right) \right)
		\d \mu(t)
	\right|
	\d \tau
\\ & \quad  \quad
+
	\ C_\nu \eps^{2\alpha}
	\int_{1/\eps}^\infty \tau^{\beta-1}
	\left|
		\int_{\eps\tau}^\infty  
		R_\omega (\eps )
		\d \mu(t)
	\right|
	\d \tau
\\ & \quad  \quad
+
	C_\nu 
	\int_{1/\eps}^\infty \tau^{\beta-1}
	\left|
	\int_{\eps \tau}^\infty 
		R_\nu \left(
			g_\eps^{-1}\left( \frac{\tau}{t}\right)
		\right)
		\mu (\d t)
	\right|
	\d \tau
\\ & \quad  \quad
+
	C_\nu 
	\int_{1/\eps}^\infty \tau^{\beta-1}
	\left|
		\int_{\eps\tau}^\infty  
		R_\nu (1/\eps )
		\d \mu(t)
	\right|
	\d \tau
\end{split}
\label{eq:part3}
\end{equation}
Noting that 
\begin{equation}
\frac{1}{\eps^2} h_0^{-1}\left( \frac{\eps^2 \tau}{t}\right)
= 
\frac{\tau}{t - \eps^2 \tau}
\ge 
\frac{\tau - t}{t - \eps^2 \tau}
=
g_\eps^{-1}\left( \frac{\tau}{t}\right)
\end{equation}
and thus
\begin{equation}
\left|
	\left(
		\frac{1}{\eps^2}
		h_0^{-1}\left( \frac{\eps^2 \tau}{t}\right)
	\right)^{-\alpha}
	-
	\left(
		g_\eps^{-1}\left( \frac{\tau}{t}\right)
	\right)^{-\alpha}
\right|\,
\le \,
\alpha 
\frac{t(t-\eps^2\tau)^\alpha}{(\tau - t)^{\alpha+1}}
\, \le \, \alpha 
\frac{t^{\alpha+1}}{(\tau - t)^{\alpha+1}}\, ,
\end{equation}
we see that
\begin{multline}
	\left|
		\int_{\eps\tau}^\infty  
		\left[
			\left(
				\frac{1}{\eps^2}
				h_0^{-1}\left( \frac{\eps^2 \tau}{t}\right)
			\right)^{-\alpha}
			-
			\left(
				g_\eps^{-1}\left( \frac{\tau}{t}\right)
			\right)^{-\alpha}
		\right] 
		\mu (\d t)
	\right|
\\
\le \, 
	\alpha 
	\int 
	\ind_{[0,t/\eps]} (\tau)
	\frac{t^{\alpha+1}}{(\tau - t)^{\alpha+1}} 
	\mu (\d t)
\, \le\,  
	\alpha c_+^{\alpha + 1}
	(\tau - c_+)^{-1-\alpha} \, ,
\label{eq:part3a}
\end{multline}
and therefore the first term on the right hand side of \eqref{eq:part3} is $\Oeps{1+\alpha-\beta}$.

The remaining terms in \eqref{eq:part3} can be estimated easily using Lemmas~\ref{lem:nu_asymptotics} and~\ref{lem:omega_asymptotics}, and they are either $\Oeps{2\alpha - \beta}$ or $\Oeps{\alpha+\delta-\beta}$.  In the second and third terms,
due to the bounded support of $\mu$ the inner integral is nonzero unless $\tau < c_+ / \eps$.  The integrands are both $\O{1}$ uniformly on the domain of integration for $\eps \searrow 0$, 
so both terms are of order
\begin{equation}
	\O{ 
		\eps^{2\alpha}
		\int_{1/\eps}^{c_+/\eps} \tau^{\beta - 1} \d \tau
	}
	= \Oeps{2 \alpha - \beta}.
	\label{eq:part3b}
\end{equation}

In the last two terms, the inner integral is again zero unless $\tau < c_+ / \eps$, but the integrand is now of order $\Oeps{\alpha+\alpha}$, so these terms are of order
\begin{equation}
	\O{
		\eps^{\alpha+\delta}
		\int_{1/\eps}^{c_+/\eps} \tau^{\beta - 1} \d \tau
	}
	=\Oeps{\alpha+\delta-\beta}.
	\label{eq:part3c}
\end{equation}
Summing up, the right hand side of \eqref{eq:part3} is $\Oeps{2\alpha-\beta}+\Oeps{\alpha+\delta-\beta}$.

The last term in \eqref{eq:big_split} is quite similar to the previous one.  We first rearrange as in \eqref{eq:part3},
\begin{equation}
\begin{split}
	&
	\int_{c_-}^{1/\eps}
	\tau^{\beta - 1}
	\bigg\vert 
	\int_0^{\eps \tau}
\\ & \times 
	\left[
	G_{\nu_0} \left( g_0^{-1} \left(\frac{\tau}{t}\right) \right)
	-
	G_{\nu_0}  \left( \frac{1}{\eps} \right)
	+
	a(\eps) G_{\omega_0} (\eps )
	-
	a(\eps) G_{\omega_0} \left( \eps^2 g_\eps^{-1} \left( \frac{\tau}{t}\right) \right)
	\right] 
	\mu (\d t) 
	\bigg\vert \d \tau
\\ &
=
	C_\nu 
	\int_{c_-}^{1/\eps}
	\tau^{\beta - 1}
	\left|
		\int_0^{\eps \tau}
		\left[
			\left(
				g_0^{-1} \left(\frac{\tau}{t}\right)
			\right)^{\alpha}
			-
			\left(
				g_\eps^{-1} \left(\frac{\tau}{t}\right)
			\right)^{\alpha}
		\right] \mu (\d t) 
	\right| \d \tau
\\ & \quad
+
	C_\nu 
	\int_{c_-}^{1/\eps}
	\tau^{\beta - 1}
	\left|
		\int_0^{\eps \tau}
		R_\nu
		\left( g_0^{-1} \left(\frac{\tau}{t}\right) \right)
		\mu (\d t)
	\right|
	\d \tau
\\ & \quad
+
	C_\nu 
	\int_{c_-}^{1/\eps}
	\tau^{\beta - 1}
	\left|
		\int_0^{\eps \tau}
		R_\nu
		\left( \frac{1}{\eps} \right)
		\mu (\d t)
	\right|
	\d \tau
\\ & \quad
+
	C_\nu \eps^{2\alpha}
	\int_{c_-}^{1/\eps}
	\tau^{\beta - 1}
	\left|
		\int_0^{\eps \tau}
		R_\omega
		\left( \eps^2 g_\eps^{-1} \left( \frac{\tau}{t}\right) \right)
		\mu (\d t)
	\right|
	\d \tau
\\ & \quad
+
	C_\nu \eps^{2\alpha}
	\int_{c_-}^{1/\eps}
	\tau^{\beta - 1}
	\left|
		\int_0^{\eps \tau}
		R_\omega
		(\eps)
		\mu (\d t)
	\right|
	\d \tau
	.
\end{split}
\label{eq:part4}
\end{equation}

Noting that
\begin{equation}
g_0^{-1}
\left(\frac{\tau}{t}\right)
=
\frac{\tau - t}{t}
<
\frac{\tau - t}{t - \eps^2 \tau}
=
g_\eps^{-1} \left(\frac{\tau}{t}\right)\, 
\end{equation}
and so
\begin{equation}
\begin{split}
	\left|
		\left(
			g_0^{-1} \left(\frac{\tau}{t}\right)
		\right)^{\alpha}
		-
		\left(
			g_\eps^{-1} \left(\frac{\tau}{t}\right)
		\right)^{\alpha}
	\right|
\le 
	\alpha
	\frac{\eps^2 t^\alpha}{(\tau - t)^\alpha (t - \eps^2 \tau)}\, ,
\end{split}
\end{equation}
which behaves like $\eps^2 \tau^{-\alpha}$ for $\tau$ on the order of $\eps^{-1}$, the first term in \eqref{eq:part4} is $\Oeps{2+\alpha - \beta}$.  The remaining terms are easy to estimate, giving either
\begin{equation}
	\O{
		\eps^{\alpha+\delta} 
		\int_0^{1/\eps} \tau^{\beta - 1} \d \tau
	}
	=\Oeps{\alpha+\delta-\beta}\, ,
	\label{eq:part4b}
\end{equation}
or
\begin{equation}
	\O{
		\eps^{2\alpha} 
		\int_0^{1/\eps} \tau^{\beta - 1} \d \tau
	}
	=\Oeps{2\alpha-\beta}\, .
	\label{eq:part4c}
\end{equation}
Combining this with the estimates for \eqref{eq:part1}, \eqref{eq:part2} and~\eqref{eq:part3} we conclude that
\begin{equation}
\vertiii{T_\eps \hat{\gamma}_\eps - \hat{\gamma}_\eps}_\beta
=
\Oeps{2\alpha - \beta}
+\Oeps{\alpha+\delta - \beta}.
\label{eq:big_bound}
\end{equation}

\appendix

\section{Some useful identies}
\label{sec:appendix}

Here are some useful identities: recall  the definitions   \eqref{eq:T_def_integral},  \eqref{eq:cumulative} and \eqref{eq:S_def_integral}.
Setting $f = \ind_{(-\infty,\sigma])}$, we obtain
\begin{equation}
F_{S_\eps \omega} (\sigma)\, 
=\, 
\int F_\omega \left( h_\eps^{-1} \left( \frac{\sigma}{z} \right) \right)  \mu(\d z)
+ F_\omega(\infty)F_\mu(\gs)\, ,
\label{eq:S_def_CDF}
\end{equation}
for all $\sigma$, and this is  another way to define the action of $S_\gep$ on $\go$.  
Of course $F_\omega(\infty)$ is one if $\go$ is a probability measure, but in general $\go$ is not normalized. 
In making use of  \eqref{eq:S_def_CDF} it is helpful to note that 
\begin{equation}
h_\eps^{-1} (y) = \frac{y-\eps^2}{1-y}.
\label{eq:h_inverse} 
\end{equation}
Observing that $h_\eps^{-1}(y) < 0 $ for $ y > 1$, we also see that
\begin{equation}
G_{S_\eps \omega} (\sigma)
=
\int_\sigma^\infty G_\omega \left(
h_\eps^{-1} \left( \frac{\sigma}{z} \right) 
\right)
\mu ( \d z)\,.
\label{eq:S_def_CDF2}
\end{equation}
Moreover
\begin{equation}
F_{T_\eps \nu}(\tau) 
\, =\,  \int F_\mu\left( \frac \tau{g_\gep(s)}\right) \nu( \dd s) \,=\, 
\int F_\nu \left( g_\eps^{-1} \left( \frac{\tau}{t} \right) \right) \mu (\d t)
+ F_\nu(\infty) F_\mu(\gep^2 \tau)\, ,
\label{eq:T_def_CDF}
\end{equation}
where
\begin{equation}
g_\eps^{-1}(y) \, =\,  \frac{y-1}{1-\eps^2 y}.
\label{eq:g_inverse}
\end{equation}
In fact, $T_{\gep} \nu$ is defined also
by the first equality in \eqref{eq:T_def_CDF}, 
or by equating  the left-most and right-most expressions.
It is useful to note that \eqref{eq:T_def_CDF} and \eqref{eq:S_def_CDF} 
can be rewritten (using $h_\gep^{-1}(y) <0$ for $y>1$ and $g_\gep^{-1}(y) <0$ for $y>1/\gep^2$)
\begin{equation}
\label{eq:alt-G}
G_{S_\gep \go}(\gs)\, =\, \int_\gs^\infty G_{\go}\left( h_\gep^{-1}\left( \frac \gs t\right)\right) \mu(\dd t)\ \textrm{ and } \
G_{T_\gep\nu}(\tau)\, =\, \int_{\gep^2 \tau}^\infty G_\nu\left( g_\gep^{-1}\left(\frac \tau t \right)\right) \mu (\dd t) \, ,
\end{equation}
for every $\gep\ge 0$.

\section*{Acknowledgements}  
  The authors wish to thank Gunter Stoltz and an anonymous referee for their comments that lead to improvements both in the presentation of the results and in some of the arguments of proof.

\end{document}